\documentclass[journal]{IEEEtran}
\usepackage{hyperref}
\usepackage{amsmath,graphicx,amssymb,amsthm}
\usepackage{cleveref}
\usepackage{pgfplotstable}
\usepackage[tworuled]{algorithm2e}
\makeatletter
\renewcommand{\@algocf@capt@plain}{above}
\makeatother
\usepackage{verbatim}

\usepackage{mathrsfs}

\graphicspath{{./}}

\newcommand{\tens}[1]{\boldsymbol{\mathcal{#1}}}
\newcommand{\tenselem}[1]{\mathcal{#1}}
\newcommand{\tY}{\tens{Y}}
\newcommand{\tG}{\tens{G}}

\newcommand{\matr}[1]{\boldsymbol{#1}}
\newcommand{\mA}{\matr{A}}
\newcommand{\mB}{\matr{B}}
\newcommand{\mC}{\matr{C}}
\newcommand{\mI}{\matr{I}}
\newcommand{\mP}{\matr{P}}
\newcommand{\mU}{\matr{U}}
\newcommand{\mV}{\matr{V}}
\newcommand{\mW}{\matr{W}}
\newcommand{\mZ}{\matr{Z}}

\newcommand{\bmx}{\begin{bmatrix}}
\newcommand{\emx}{\end{bmatrix}}
\newcommand{\bsm}{\left[\begin{smallmatrix}}
\newcommand{\esm}{\end{smallmatrix}\right]}

\newcommand{\vect}[1]{\boldsymbol{#1}}

\newcommand{\mlprod}[4]{[\![#1;\, #2,#3,#4]\!]}
\newcommand{\cpd}[3]{[\![#1, #2,#3]\!]}

\newcommand{\kr}{\odot}     
\newcommand{\kron}{\mathop{\boxtimes}}   
\newcommand{\con}{\mathop{\bullet}}   
\newcommand{\T}{{\sf T}}        

\makeatletter
\newcommand{\vecl}[1]{\mathop{\operator@font vec}\{#1\}}
\newcommand{\rank}[1]{\mathop{\operator@font rank}\{#1\}}
\newcommand{\colrank}[1]{\mathop{\operator@font colrank}\{#1\}}
\newcommand{\krank}[1]{\mathop{\operator@font krank}\{#1\}}
\newcommand{\trace}[1]{\mathop{\operator@font trace}\{#1\}}
\newcommand{\Diag}[1]{\mathop{\operator@font Diag}\{#1\}}    
\newcommand{\diag}[1]{\mathop{\operator@font diag}\{#1\}}    
\newcommand{\Span}[1]{\mathop{\operator@font Span}\{#1\}}    
\newcommand{\argmin}{\mathop{\operator@font argmin}}

\newcommand{\cond}[1]{\mathop{\operator@font cond}\{#1\}}
\newcommand{\floor}[1]{\lfloor #1 \rfloor}
\makeatother

\newcommand{\RR}{\mathbb{R}}

\newcommand{\unfold}[2]{\matr{#1}^{(#2)}}
\newcommand{\contr}[1]{\con_{#1}}

\newtheorem{definition}{Definition}[section]
\newtheorem{theorem}[definition]{Theorem}
\newtheorem{proposition}[definition]{Proposition}

\newtheorem{corollary}[definition]{Corollary}
\newtheorem{remark}[definition]{Remark}

\newcommand{\OurAlgo}{{SCOTT} }

\begin{document}
\title{Hyperspectral Super-Resolution\\ with Coupled Tucker Approximation: \\ Recoverability and SVD-based algorithms}

\author{Cl{\'e}mence~Pr{\'e}vost,~\IEEEmembership{Student Member,~IEEE,}
Konstantin~Usevich$^*$, ~\IEEEmembership{Member,~IEEE,}\\
Pierre~Comon,~\IEEEmembership{Fellow,~IEEE,} 
and~David~Brie~\IEEEmembership{Member,~IEEE,}%
\thanks{This  work  was  partially  supported  by  the  ANR (Agence Nationale de Recherche) under  grants OPTIFIN (ANR-15-CE10-0007) and LeaFleT (ANR-19-CE23-0021).}%
\thanks{C. Pr{\'e}vost, K. Usevich and D. Brie are with Centre de Recherche en Automatique de Nancy (CRAN), Universit{\'e} de Lorraine, CNRS, Boulevard des Aiguillettes, BP 70239, F-54506 Vandoeuvre-l{\`e}s-Nancy, France. 
P. Comon is with  CNRS, GIPSA-Lab, Univ.  Grenoble Alpes, F-38000 Grenoble, France.}
\thanks{Email: clemence.prevost@univ-lorraine.fr, konstantin.usevich@univ-lorraine.fr, pierre.comon@gipsa-lab.fr, david.brie@univ-lorraine.fr. Fax: +33 383684437. Tel.:~+33 372745313 (K. Usevich).}%
\thanks{*Corresponding author.}}%
\date{ }

\maketitle

\maketitle

\begin{abstract}
We propose a novel  approach for hyperspectral super-resolution, that is based on low-rank tensor approximation for a coupled low-rank multilinear (Tucker) model.
We show that the  correct recovery holds  for a wide range of multilinear ranks. 
For coupled tensor approximation, we propose  two SVD-based algorithms  that are  simple and fast, but with a performance comparable to the state-of-the-art methods. 
The approach is applicable to the case of unknown spatial degradation and  to the pansharpening problem.
\end{abstract}

\section{Introduction}
\label{sec:intro}

The problem of \emph{hyperspectral super-resolution} (HSR) \cite{YokoyaGC17:hsr} has recenlty attracted much interest from the signal processing community.
It consists in fusing a multispectral image (MSI), which has a good spatial resolution but few spectral bands, and a hyperspectral image (HSI), whose spatial resolution is lower than that of  MSI.
The aim  is to recover a \emph{super-resolution image} (SRI), which possesses both good spatial and spectral resolutions.
This problem is closely related to hyperspectral pansharpening \cite{LoABB15:pansharpening,Aiazzi11:hsr}, where  the  HSI is fused with a panchromatic image (i.e. an ``MSI'' with  one spectral band).

Many methods were developed for the HSR problem, including coupled nonnegative matrix factorization \cite{YokoYI12:cnmf} (CNMF), methods based on solving Sylvester equations \cite{WeiBT15:hsr}, Bayesian approaches (HySure \cite{SimoesBAC15:hsr}), FUMI \cite{WeiBT16:hsr}, to name a few.  Motivated by the widely used linear mixing model,  most of these methods are based  on  a coupled low-rank factorization of the matricized hyperspectral and multispectral images. 
In \cite{li18:hsr}, a matrix factorization approach under sparsity conditions was proposed, together with a proof of correct recovery of the estimated SRI in the noiseless case.

Recently, a promising tensor-based method was proposed that makes use of the inherent 3D nature of HSI \cite{KanatsoulisFSM:hsr}.
Assuming that the super-resolution image  itself  admits a low-rank canonical polyadic (CP) decomposition (CPD), the  HSR is reformulated as a coupled CP  approximation.
 An alternating least squares (ALS) algorithm called 
Super-resolution TEnsor REconstruction (STEREO)  is proposed, achieving reconstruction performance that is competitive with the state of the art.
A proof of the correct recovery of the SRI by the approach of \cite{KanatsoulisFSM:hsr} is given provided the CPD of the MSI is unique. 
 This approach was also successfully used  for a super-resolution problem in medical imaging \cite{HatvaniBTGK18:cpd}.

In some cases, the spatial degradation operator is unknown, therefore blind algorithms are needed.
A blind version of STEREO  was proposed in \cite{KanatsoulisFSM:hsr} that also uses an ALS algorithm for a coupled CP model.
In \cite{KanaFSM18:scuba}, a simple Super-resolution CUBe Algorithm  (SCUBA) based on a single CPD of the MSI tensor and a truncated SVD of the unfolding of the HSI is introduced.
A key idea proposed in \cite{KanaFSM18:scuba} is to use local approximations by splitting the data cubes into separate blocks.
This algorithm outperforms blind STEREO and other state-of-the-art algorithms.
It also does not require separability of the spatial degradation operator.

In this paper, we propose  to use another type of low-rank tensor factorization: multilinear (also known as Tucker) factorization. 
By assuming that the super-resolution image  has approximately low multilinear rank, we reformulate the HSR problem as a coupled Tucker approximation.
 First, we propose   two  closed-form SVD-based algorithms: the first, named Super-resolution based on COupled Tucker Tensor approximation (SCOTT),  is inspired by the higher-order SVD \cite{DeLaDMV00:hosvd} and the second (blind) is inspired by \cite{KanaFSM18:scuba}.
Second, we prove that,  although the Tucker decomposition is not identifiable
, the SRI can be uniquely recovered for a wide range of multilinear ranks.
While the proposed exact recovery conditions are in general more restrictive than those of \cite{KanatsoulisFSM:hsr}, they can be specialized in situations for which nothing can be concluded from \cite{KanatsoulisFSM:hsr}. 
Our experiments on a number of simulated and semi-real examples, show that the proposed  algorithms have  a performance approaching  those of \cite{KanatsoulisFSM:hsr} and \cite{KanaFSM18:scuba}, but the computational cost is much lower.
Also, the proposed approach is applicable to hyperspectral pansharpening \cite{He14:pansharpening}  (unlike \cite{KanatsoulisFSM:hsr}, which requires the MSI to have at least two spectral bands).
Finally, the algorithms can accurately reconstruct spectral signatures,  which is of prime importance for  further processing of the HSR image.

A  short version of this work \cite{PrevostUCB:icassp} appears in ICASSP 2019, presenting the  SCOTT algorithm and part of the simulations. The current  paper additionally includes new blind algorithms, detailed analysis of the model and the algorithms, proof of the theorem for recoverability,  new simulations  for  synthetic and semi-real data, examples on recovery of spectral signatures.

This paper is organized as follows.
In  \Cref{sec:notation}, we introduce our notation, define basic tensor decomposition operations and recall the HSR problem.
In \Cref{sec:cpd}, we recall the  CP-based model and the STEREO algorithm proposed in \cite{KanatsoulisFSM:hsr}.
\Cref{sec:tucker} contains our proposed coupled Tucker model and SVD-based algorithms (SCOTT and B-SCOTT) for tensor approximation.
In \Cref{sec:id} we prove our main recoverability result  for the coupled Tucker model.
\Cref{sec:exp} contains the numerical experiments. 

\section{Background and notation} \label{sec:notation}
\subsection{Basic notation}
In this paper we mainly follow \cite{Comon14:spmag} in what concerns the tensor notation (see also \cite{KoldaB09:sirev}).
The following fonts are used:  lowercase  ($a$) or uppercase ($A$)  plain  font for scalars,   boldface lowercase ($\vect{a}$) for vectors,  uppercase boldface ($\matr{A}$) for matrices, and calligraphic ($\tens{A}$) for $N$-D arrays (tensors). Vectors are, by convention, one-column matrices. The elements of vectors/matrices/tensors are 
accessed as $a_{i} $, ${A}_{i,j}$ and $\tenselem{A}_{i_1,\ldots,i_N}$ respectively.
$\RR$ stands for the real line.

For a matrix  ${\mA}$, we denote its transpose and Moore-Penrose pseudoinverse as ${\mA}^{\T}$ and ${\mA}^\dag$ respectively.  The notation ${\mI}_M$ is used for the $M\times M$ identity matrix and $\matr{0}_{L \times K}$ for the $L\times K$ matrix of zeroes.
We use  the symbol $\kron$ for the Kronecker product of  matrices (in order to distinguish it from the  tensor  product $\otimes$), and $\kr$ for the Khatri-Rao  product.

For a matrix $\matr{X}\in\RR^{m\times n}$, we denote by $\sigma_{max}(\matr{X})$ and $\sigma_{min}(\matr{X})$ the largest and the smallest of the $\min(m,n)$ singular values of $\matr{X}$.
We also denote by $\text{tSVD}_{R}\left(\matr{X}\right)\in\RR^{n\times R}$  a matrix containing  $R$ leading  right  singular vectors of  $\matr{X}$.

We use $\vecl{\cdot}$ for the standard column-major vectorization of a tensor or a matrix. 
Operator $\contr{p}$ denotes contraction on the $p$th index of a tensor; when contracted with a matrix, summation is always performed on the second index of the matrix, e.g., $[\tens{A}\contr{1}\matr{M}]_{ijk}=\sum_\ell \tenselem{A}_{\ell jk} M_{i\ell}$.
For a tensor $\tens{Y} \in \RR^{I \times J \times K}$, its first unfolding is denoted by $\unfold{Y}{1} \in \RR^{JK \times I}$.

\subsection{Tensor decompositions}
For a tensor $\tens{G} \in \RR^{R_1 \times R_2 \times R_3}$ and matrices ${\mU} \in \RR^{I \times R_1}$, ${\mV}\in \RR^{J \times R_2}$ and ${\mW}\in \RR^{K \times R_3}$,  the following shorthand notation is used for the multilinear product:
\begin{equation}\label{eq:}
\mlprod{\tens{G}}{{\mU}}{{\mV}}{{\mW}} = \tens{G} \con_1 {{\mU}} \con_2 {{\mV}} \con_3 {{\mW}}.
\end{equation}
which means that the $(i,j,k)$th entry of the above array is
\[
\sum_{pqr} G_{pqr} \; U_{ip} V_{jq} W_{kr}.\]
If $\tY = \mlprod{\tG}{{\mU}}{{\mV}}{{\mW}}$, the following identities hold for its vectorization and unfoldings, respectively:
\[
\begin{split}
\vecl{\tY} &= (\mW \kron \mV \kron \mU) \vecl{\tens{G}},\\
\unfold{Y}{1} &= \left(\mW\kron\mV\right)\unfold{G}{1}{\mU}^{\T}.
\end{split}
\]
If, in addition,
\[
R_1 = \rank{\unfold{Y}{1}}, R_2 = \rank{\unfold{Y}{2}}, R_3 = \rank{\unfold{Y}{3}},
\]
then the multilinear product is called Tucker decomposition of $\tY$  and $(R_1,R_2,R_3)$ are called the multilinear ranks.

For matrices $\mA \in \RR^{I \times F}$, $\mB \in \RR^{J \times F}$, $\mC \in \RR^{K \times F}$, we will use a shorthand notation for a polyadic decomposition (sometimes also called rank-decomposition)
\[
\cpd{{\mA}}{{\mB}}{{\mC}}  = \mlprod{\tens{I}_F}{\mA}{\mB}{\mC},
\]
where $\tens{I}_F \in \RR^{F\times F \times F}$ is a diagonal tensor of ones. 
 In other words, if $\tens{Y}=\cpd{{\mA}}{{\mB}}{{\mC}}$, then
\begin{eqnarray*}
\tenselem{Y}_{ijk} = \sum_r A_{ir} B_{jr} C_{kr}; 
\end{eqnarray*}
moreover, the first unfolding   can be expressed as
\begin{eqnarray*}
\unfold{Y}{1} = (\matr{C} \kr \matr{B} ) \matr{A}^\T.
\end{eqnarray*} 
Finally, if $F$ is minimal,  $\tens{Y}=\cpd{{\mA}}{{\mB}}{{\mC}}$ is called canonical polyadic (CP) decomposition and  $F$ is called the tensor rank.

\subsection{Hyperspectral super-resolution and degradation model}\label{sec:HSR}
We consider a multispectral image (MSI) cube ${\tY}_M \in \RR^{I\times J \times K_M}$ and a hyperspectral image (HSI) cube ${\tY}_H \in \RR^{I _H\times J_H \times K}$ acquired from existing sensors (for instance, LANDSAT or QuickBird). 
The spectral resolution of the  MSI is lower than that of the  HSI ($K_M \ll K$) , while its spatial resolution is higher ($I> I_H$, $J > J_H$).
The acquired MSI and HSI usually represent the same target, and ${\tY}_M$ and ${\tY}_H$ are viewed as two degraded  versions   of a single super-resolution image (SRI) data cube ${\tY} \in \RR^{I\times J \times K}$.
The hyperspectral data fusion problem \cite{YokoyaGC17:hsr}  consists in recovering  SRI ${\tY}$  from  ${\tY}_M$ and ${\tY}_H$.  
 
In this paper, as in \cite{KanatsoulisFSM:hsr}, we adopt the following degradation model, that can be compactly written as contraction of SRI:
\begin{equation}\label{eq:lmm}
\begin{cases}
{\tY}_M &= {\tY} \con_3 {\matr{P}_M} + \tens{E}_M, \\
{\tY}_H &= {\tY} \con_1 {\matr{P}_1} \con_2 {\matr{P}_2} + \tens{E}_H,
\end{cases}
\end{equation}
where $\tens{E}_M$, $\tens{E}_H$ denote the noise terms, $\matr{P}_M \in \RR^{K_M \times K}$  is the spectral degradation matrix (for example, a selection-averaging matrix),  and $\matr{P}_1\in \RR^{I_H \times I}$, $\matr{P}_2\in \RR^{J_H \times J}$   are the spatial degradation matrices, \emph{i.e.} we assume (for simplicity) that the spatial degradation is separable.
For example, the commonly accepted Wald's protocol \cite{WaldRM97:hsr}  uses separable Gaussian blurring and downsampling in both spatial dimensions.

\section{CP-based data fusion}
\label{sec:cpd}
In \cite{KanatsoulisFSM:hsr} it was proposed to model the SRI data cube as a tensor with low tensor rank, i.e. ${\tY} = \cpd{\mA}{\mB}{\mC}$, where $\matr{A} \in \RR^{I\times F}$, $\matr{B} \in \RR^{J\times F}$ and $\matr{C} \in \RR^{K\times F}$ are the factor matrices of the CPD and  $F$ is the tensor rank. 
\subsection{The case of known spatial degradation (STEREO)}
In this  subsection, we consider only the case when the degradation matrices $\matr{P}_1$, $\matr{P}_2$, $\matr{P}_M$ are known; the case of unknown degradation matrices is postponed to Section~\ref{unknown-CP:sec}.
In this case, the HSR problem can be formulated as
\begin{equation}\label{eq:coupled_cpd}
\underset{\widehat{\mA}, \widehat{\mB}, \widehat{\mC}}{\text{minimize}}\quad  f_{ CP}(\widehat{\mA},\widehat{\mB},\widehat{\mC}), 
\end{equation}
where $f_{ CP}(\widehat{\mA},\widehat{\mB},\widehat{\mC})=$
\[
\| {\tY}_H - \cpd{\matr{P}_1\widehat{\mA}}{\matr{P}_2\widehat{\mB}}{\widehat{\mC}}\|_F^2  + \lambda \| {\tY}_M - \cpd{\widehat{\mA}}{\widehat{\mB}}{\matr{P}_M\widehat{\mC}}\|_F^2,
\]
which is a coupled CP approximation problem.
As in \cite{KanatsoulisFSM:hsr} we set $\lambda = 1$ so that both degradated images have the same weight in the cost function.
Thus, we consider that the HSI and MSI share the same level of additive noise. 
For the  case when there is no noise ($\tens{E}_{H}, \tens{E}_M = \matr{0}$), the coupled CP model is (generically) identifiable if 
\begin{equation*}
F\leq \min\{ 2^{\floor{\log_2(K_MJ)}-2}, I_HJ_H \},
\end{equation*}
see \cite{KanatsoulisFSM:hsr} for a proof and details of this condition.

To solve \eqref{eq:coupled_cpd}, an alternating optimization algorithm is proposed in \cite{KanatsoulisFSM:hsr}, called 
STEREO
, as it is described in Algorithm~\ref{alg:stereo}.

\begin{algorithm}[ht!]
\SetKwInOut{Input}{input}\SetKwInOut{Output}{output}
 			\Input{${\tY}_M$, ${\tY}_H$, $\matr{P}_1$, $\matr{P}_2$, $\matr{P}_M$; $F$, $\mA_0\in \RR^{I\times F}$, $\mB_0\in \RR^{J\times F}$, $\mC_0\in \RR^{K\times F}$}
 			\Output{$\widehat{\tY} \in \RR^{I\times J\times K}$}
			\For{$k=1:n$}{
${\mA}_k \leftarrow \underset{\mA}{\argmin} \,  f_{ CP}({\mA},{\mB}_{k-1},{\mC}_{k-1})$,\\
${\mB}_k \leftarrow \underset{\mB}{\argmin} \, f_{CP}({\mA}_{k},{\mB},{\mC}_{k-1})$,\\
${\mC}_k \leftarrow \underset{\mC}{\argmin}  \, f_{CP}({\mA}_k,{\mB}_{k},{\mC})$,}
	 $\widehat{\tY} \leftarrow  \cpd{{\mA}_n}{{\mB}_n}{{\mC}_n}$.
\caption{STEREO}
\label{alg:stereo}
\end{algorithm}

The updates of the factor matrices in Algorithm~\ref{alg:stereo} can be computed by using efficient  solvers for the (generalized) Sylvester equation \cite{BartelsS72:syl}, \cite{GolubNV79:syl}.
For example, the total cost of one iteration (updating $\mA,\mB,\mC$) in Algorithm~\ref{alg:stereo}  becomes 
\begin{itemize}
\item $O(IJK_M F + I_HJ_HK F )$ flops for computing the right hand sides in the  least-squares subproblems.
\item $O(I^3\!+\!J^3\!+\!K^3\!+\!F^3)$ flops for solving Sylvester equations;
\end{itemize}
For more details on solving Sylvester equations, see\footnote{Note that in \cite[Appendix~E]{KanatsoulisFSM:hsr} the cost of solving the Sylvester equation is stated as $O(I^3)$ and not $O(I^3 +F^3)$ as in \cite{BartelsS72:syl}.} \cite[Appendix~E]{KanatsoulisFSM:hsr} and Appendix~\ref{app:1}. 
The initial  values  in Algorithm~\ref{alg:stereo} are chosen as in Algorithm~\ref{alg:tenrec}:

\begin{algorithm}[ht!]
\SetKwInOut{Input}{input}\SetKwInOut{Output}{output}
 			\Input{${\tY}_M$, ${\tY}_H$, $\matr{P}_1$, $\matr{P}_2$; $F$}
 			\Output{$\mA_0\in \RR^{I\times F}$, $\mB_0\in \RR^{J\times F}$, $\mC_0\in \RR^{K\times F}$}
			 $\cpd{\mA_0}{\mB_0}{\matr{\widetilde{C}}_0} = \text{CPD}_F \left({\tY}_M\right)$,\\
			 $\mC_0^{\T} =  \left(\matr{P}_2\matr{B}_0\kr \matr{P}_1\matr{A}_0\right)^{\dagger}{\unfold{Y}{3}_H }$.

\caption{TenRec}
\label{alg:tenrec}
\end{algorithm}

where $\text{CPD}_F(\tY_M)$ stands for a rank-$F$ CP approximation\footnote{A low tensor rank approximation does not always exist in general, but is guaranteed to exist if all terms are imposed to be entry-wise nonnegative \cite{QiCL16:semialgebraic}.} of $\tY_M$, and $\mC_0$ is obtained by solving a least-squares problem.
Algorithm~\ref{alg:tenrec} can be used as an algebraic method for solving the HSR problem.

\subsection{The case of unknown spatial degradation}\label{unknown-CP:sec}
In this subsection, we recall the CP-based methods for the HSR problem in the case when  the spatial degradation matrices $\matr{P}_1$, $\matr{P}_2$ are unknown,   proposed in  \cite{KanatsoulisFSM:hsr} and \cite{KanaFSM18:scuba}.
The first solution, called Blind-STEREO was to consider the following coupled CP approximation problem:
\[
\underset{\substack{\widehat{\mA}, \widehat{\mB}, \widehat{\mC}\\ \widetilde{\mA}, \widetilde{\mB}}}{\min} \| {\tY}_H - \cpd{\widetilde{\mA}}{\widetilde{\mB}}{\widehat{\mC}}\|_F^2  + \lambda \| {\tY}_M - \cpd{\widehat{\mA}}{\widehat{\mB}}{\matr{P_M}\widehat{\mC}}\|_F^2,
\]
where the estimated SRI is computed as $\widehat{\tY} = \cpd{\widehat{\mA}}{\widehat{\mB}}{\widehat{\mC}}$ and the matrices $\widetilde{\mA}$, $\widetilde{\mB}$ represent degraded versions of $\widehat{\mA}$, $\widehat{\mB}$ by unknown spatial degradation matrices, respectively. 
The conditions for correct recovery were established in  \cite{KanatsoulisFSM:hsr}.

In \cite{KanaFSM18:scuba}, an alternative approach was proposed, that uses standard CP approximation of $\tens{Y}_M$ together with an SVD of $\unfold{Y}{3}_H$, and a least squares problem. 
This approach, which does not necessary need separability of the spatial degradation operation,  is summarized in Algorithm~\ref{alg:hybrid}.

\begin{algorithm}[ht!]
\SetKwInOut{Input}{input}%
\SetKwInOut{Output}{output}
			\Input{$\tY_M$, $\tY_H$, $\mP_M$; $R$, $F$}
 			\Output{$\widehat{\tY} \in \RR^{I\times J\times K}$}
			Compute CP approximation: $\cpd{\widehat{\mA}}{\widehat{\mB}}{\widetilde{\mC}} = \text{CPD}_F (\tens{Y}_M)$, \\
			$\mZ \leftarrow \text{tSVD}_{R}\left(\unfold{Y}{3}_H\right)$, \\
			$\widehat{\mC} \leftarrow \mZ (\mP_M \mZ)^{\dagger} \widetilde{\mC}$, \\
		       $\widehat{\tY} \leftarrow  \cpd{\widehat{\mA}}{\widehat{\mB}}{\widehat{\mC}}$.
		      
			\caption{Hybrid algorithm of \cite{KanaFSM18:scuba}}
			 \label{alg:hybrid}
		\end{algorithm}

As noted in \cite{KanaFSM18:scuba},  ${\tY} =  \cpd{{\mA}}{{\mB}}{{\mC}}$ can be uniquely recovered only if  $\rank{\mC} = R$ does not exceed the number $K_M$ of spectral bands in the  MSI.
To overcome this limitation, in  \cite{KanaFSM18:scuba} it was proposed to apply Algorithm~\ref{alg:hybrid} to corresponding  non-overlapping subblocks of the  MSI and HSI (based on the hypothesis that only a small number of materials are active in a smaller block).
This is summarized in Algorithm~\ref{alg:scuba}, called SCUBA in \cite{KanaFSM18:scuba}.
It was shown in \cite{KanaFSM18:scuba} that such an algorithm outperforms Blind-STEREO, and other state-of-the-art algorithms for blind HSR.

\begin{algorithm}[ht!]
\SetKwInOut{Input}{input}\SetKwInOut{Output}{output}
 			\Input{$\tY_M$, $\tY_H$, $\mP_M$; $R$, $F$}
 			\Output{$\widehat{\tY} \in \RR^{I\times J\times K}$}
			Split  $\tY_M$, $\tY_H$ in $L$  blocks along  spatial dimensions. 
			
			\For{$k=1:L$}{
			Apply Algorithm~\ref{alg:hybrid} to each pair of blocks in $\tY_M$, $\tY_H$, and store the result in the corresponding block of $\widehat{\tY}$.
			}
			\caption{SCUBA}
			 \label{alg:scuba}
		\end{algorithm}

\section{Tucker-based data fusion}
\label{sec:tucker}

\subsection{Model and approximation problem}
In this paper, we propose a Tucker-based coupled model\footnote{Note that another method based on Tucker factorization and sparse approximation  was proposed in \cite{DianFL17:tuck}.
However, no recoverability condition is available for that method.} as an alternative to STEREO.
Let  $\vect{R} = (R_1,R_2,R_3)$ be the multilinear ranks of  the SRI ${\tY}$, and let $\tY = \mlprod{\tens{G}}{{\mU}}{{\mV}}{{\mW}}$ be its Tucker decomposition,  where $\matr{U} \in \RR^{I\times R_1}$, $\matr{V} \in \RR^{J\times R_2}$ and $\matr{W } \in \RR^{K\times R_3}$ are the factor matrices  and $\tens{G} \in \RR^{R_1\times R_2\times R_3}$ is  the  core tensor.

 With these notations,  Equation   \eqref{eq:lmm} becomes
\begin{equation}\label{eq:HSR_Tucker}
\begin{cases}
{\tY}_M &= \mlprod{\tens{G}}{{\mU}}{{\mV}}{\matr{P_M}{\mW}} + \tens{E}_M, \\
{\tY}_H &= \mlprod{\tens{G}}{\matr{P_1}{\mU}}{\matr{P_2}{\mV}}{{\mW}} + \tens{E}_H,
\end{cases}
\end{equation}
thus the HSR task can be performed by estimating\footnote{Note that our final goal is not parameter estimation, but rather the recovery of the SRI tensor. We use the word ``estimating'' just to underscore that it is not necessary to formulate the recovery problem as an optimization problem.} the factor matrices $\mU$, $\mV$, $\mW$ and the core tensor $\tens{G}$ in the Tucker decomposition of the SRI. 

As in \Cref{sec:cpd}, one of the possible ways is to reformulate the HSR problem as an optimization problem:
\begin{equation}\label{eq:HSR_Tucker_cost}
\underset{\widehat{\tG}, \widehat{\mU}, \widehat{\mV}, \widehat{\mW} }{\text{minimize}}\;  f_{T}(\widehat{\tG}, \widehat{\mU}, \widehat{\mV}, \widehat{\mW}),  \quad \mbox{where}
\end{equation}
\begin{equation}\label{eq:HSR_Tucker_cost_fn}
\begin{split}
f_{T}(\widehat{\mU}, \widehat{\mV}, \widehat{\mW}, \widehat{\tG})   = & \| {\tY}_H - \mlprod{\widehat{\tens{G}}}{\matr{P_1}\widehat{\mU}}{\matr{P_2}\widehat{\mV}}{\widehat{\mW}}\|_F^2   \\
  + &\lambda \| {\tY}_M - \mlprod{\widehat{\tens{G}}}{\widehat{\mU}}{\widehat{\mV}}{\matr{P_M}\widehat{\mW}}\|_F^2.
\end{split}
\end{equation}
Rather than finding a (local) minimum of \eqref{eq:HSR_Tucker_cost},
we propose two  (semi-algebraic) closed-form solutions that are suboptimal, but are fast and easy to calculate.

\subsection{An  SVD-based algorithm for known spatial degradation}
A two-stage approach  inspired by the high-order SVD (HOSVD)  \cite{DeLaDMV00:hosvd}  consists in:
\begin{itemize}
\item using the truncated SVD of MSI and HSI  to obtain the factors $\widehat{\matr{U}}$, $\widehat{\matr{V}}$, $\widehat{\matr{W}}$ in \eqref{eq:HSR_Tucker};
\item performing the data fusion by minimizing the objective \eqref{eq:HSR_Tucker_cost_fn} only with respect to the core tensor $\widehat{\tens{G}}$.
\end{itemize}
This method, called SCOTT, is given in Algorithm~\ref{alg:hosvd}.

Note that, under  conditions provided in  \Cref{sec:id},   Algorithm~\ref{alg:hosvd} gives a solution to the algebraic decomposition problem \eqref{eq:HSR_Tucker} in the noise-free case ($\tens{E}_M = \matr{0}$, $\tens{E}_H = \matr{0}$). 
 \begin{algorithm}[ht!]
\SetKwInOut{Input}{input}\SetKwInOut{Output}{output}
 			\Input{$\tY_M$, $\tY_H$, $\mP_1$, $\mP_2$, $\mP_M$; $(R_1,R_2,R_3)$,}
 			\Output{$\widehat{\tY} \in \RR^{I\times J\times K}$}
			1. $\widehat{\matr{U}}  \leftarrow  \text{tSVD}_{R_1}\left(\unfold{Y}{1}_M\right)$, $\widehat{\matr{V}}  \leftarrow  \text{tSVD}_{R_2}\left(\unfold{Y}{2}_M\right)$, $\widehat{\matr{W}}  \leftarrow  \text{tSVD}_{R_3}\left(\unfold{Y}{3}_H\right)$,\\
			2. $\widehat{\tG}  \leftarrow \underset{\tG}{\argmin} \; f_{T}\left({\tG},\widehat{\mU}, \widehat{\mV}, \widehat{\mW}\right)$, \\
		        3. $\widehat{\tY} = \mlprod{\widehat{\tens{G}}}{\widehat{{\mU}}}{\widehat{{\mV}}}{\widehat{{\mW}}}$.
				
 			\caption{SCOTT}
			 \label{alg:hosvd}
		\end{algorithm}

Step 2 of Algorithm~\ref{alg:hosvd} is the least squares problem
\[ 
\underbrace{\bmx \matr{\widehat{W}}\kron\matr{P}_2\matr{\widehat{V}}\kron\matr{P}_1\matr{\widehat{U}} \\
\sqrt{\lambda}\matr{P}_M\matr{\widehat{W}}\kron\matr{\widehat{V}}\kron\matr{\widehat{U}} \emx}_{\matr{X}} \vecl{\tens{\widehat{G}}} \approx \underbrace{\bmx \vecl{{\tY}_H} \\\sqrt{\lambda} \vecl{{\tY}_M} \emx}_{\vect{z}}
\]
that can  be  solved through  normal equations of the form  
\begin{equation}\label{eq:ls_core}
\left(\matr{X}^{\T}\matr{X}\right)\vecl{\tens{\widehat{G}}} = \matr{X}^{\T}\vect{z}.
\end{equation}
The matrix on the left-hand side of (\ref{eq:ls_core}) can be written as
\begin{equation}\label{eq:left}
 \begin{split}
 \matr{X}^{\T}\matr{X} = \mI_{R_3}\kron \left(\widehat{\mV}^{\T}\mP_2^{\T}\mP_2\widehat{\mV}\right)\kron\left(\widehat{\mU}^{\T}\mP_1^{\T}\mP_1\widehat{\mU}\right) \\
 + {\lambda}\left(\widehat{\mW}^{\T}\mP_M^{\T}\mP_M\widehat{\mW}\right)\kron\mI_{R_1R_2},
 \end{split}
\end{equation}
and the vector on the right-hand side is  
\begin{equation}\label{eq:right}
 \begin{split}
 \matr{X}^{\T}\vect{z} = \vecl{\mlprod{\tY_H}{\widehat{\matr{U}}^{\T}\!\matr{P}_1^{\T}}{\widehat{\matr{V}}^{\T}\!\matr{P}_2^{\T}}{\widehat{\mW}^{\T}}} \\ + \lambda\vecl{\mlprod{\tY_M}{\widehat{\mU}^{\T}}{\widehat{\mV}^{\T}}{\widehat{\mW}^{\T}\!\matr{P}_M^{\T}}}.
  \end{split}
\end{equation}
The normal equations can be  viewed as a (generalized) Sylvester equation and  (as in the case of STEREO) efficient solvers can be used (see Appendix~\ref{app:1} for more details). 
Thus the total cost of SCOTT algorithm  becomes
\begin{itemize}
\item $O(\min(R_1,R_2) I J K_M + R_3 I_H J_H K)$ flops for computing the truncated SVDs  and computing $\matr{X}^{\T}\vect{z}$;
\item $O(\min(R_3^3 +(R_1R_2)^3, R_1^3 +(R_2R_3)^3))$ flops for solving the Sylvester equation.
\end{itemize}
It is easy to see that the computational complexity of SCOTT is comparable to that of one iteration of STEREO  and can be  smaller if the multilinear ranks are small.

\subsection{An  algorithm for unknown spatial degradation}\label{unknown-Tucker:sec}
In this subsection, we show that is also possible to develop a blind SVD-based algorithm, in the same spirit as Algorithm~\ref{alg:hybrid}.
The algorithm does not need knowledge of $\mP_1$, $\mP_2$ and is based on the HOSVD of the MSI tensor.
 
\begin{algorithm}[ht!]
\SetKwInOut{Input}{input}\SetKwInOut{Output}{output}
 			\Input{$\tY_M$, $\tY_H$, $\mP_M$; $(R_1,R_2,R_3)$}
 			\Output{$\widehat{\tY} \in \RR^{I\times J\times K}$}
1. Compute the $(R_1,R_2,R_3)$ HOSVD   of $\tY_M $
\[
\mlprod{\widehat{\tens{G}}}{\widehat{{\mU}}}{\widehat{{\mV}}}{\widetilde{{\mW}}} \stackrel{\textrm{HOSVD}}{\approx} \tY_M, 
\]\\
2. $\mZ \leftarrow  \text{tSVD}_{R_3}\left(\unfold{Y}{3}_H\right)$,\\%
3. $\widehat{\matr{W}} \leftarrow \mZ (\mP_M \mZ)^\dagger \widetilde{{\mW}}$, \\%
4. $\widehat{\tY} = \mlprod{\widehat{\tens{G}}}{\widehat{{\mU}}}{\widehat{{\mV}}}{\widehat{{\mW}}}$.
 			\caption{Blind version of SCOTT}
			 \label{alg:hosvd_blind}
		\end{algorithm}

The total computational complexity of Algorithm~\ref{alg:hosvd_blind} is
\[
O\left(\min(R_1,R_2) I J K_M + R_3 I_H J_H K\right) \text{ flops}
\]
and is dominated by the cost of the truncated SVD, because step 3 is very cheap.
However, a specific drawback of Algorithm~\ref{alg:hosvd_scuba}, similarly to Algorithm~\ref{alg:hybrid}, is that $R_3$ should not exceed $K_M$, since the multilinear rank is employed in the HOSVD of subblocks of $\tY_M$.

Finally, similarly to SCUBA, we can use a block version of Algorithm~\ref{alg:hosvd_blind}, which we call B-SCOTT (which stands for ``Blind SCOTT'').
There is no confusion, as Algorithm~\ref{alg:hosvd_blind} is a special case of Algorithm~\ref{alg:hosvd_scuba} where the degraded image cubes are not split into blocks.
\begin{algorithm}[ht!]
\SetKwInOut{Input}{input}\SetKwInOut{Output}{output}
 			\Input{$\tY_M$, $\tY_H$, $\mP_M$; $(R_1,R_2,R_3)$}
 			\Output{$\widehat{\tY} \in \RR^{I\times J\times K}$}
					Split  $\tY_M$, $\tY_H$ in $L$  blocks along  spatial dimensions. 
	
			\For{$k=1:L$}{
			Apply Algorithm~\ref{alg:hosvd_blind} to each pair of blocks in $\tY_M$, $\tY_H$, and store the result in the corresponding block of $\widehat{\tY}$.
			}
		      
			\caption{B-SCOTT (block version of Algorithm~\ref{alg:hosvd_blind})}
			 \label{alg:hosvd_scuba}
		\end{algorithm}

\section{Recoverability of the Tucker model}
\label{sec:id}
In this section, we establish  conditions for correct SRI tensor  recovery in the coupled Tucker model. 
The proof of such conditions for the CP model in \cite{KanatsoulisFSM:hsr} relied on the uniqueness (identifiability) property of the CPD of the MSI. 
We show that, although the Tucker decomposition is not unique, the correct recovery is still possible.
Moreover, we prove that in some  cases where the CPD in \cite{KanatsoulisFSM:hsr} is not unique, the SRI tensor is still uniquely recovered using the CP model.

\subsection{Deterministic exact recovery conditions}
We begin with a  deterministic result on recoverability\footnote{In this paper, we prefer to use the term ``recoverability of the SRI'' rather than ``identifiability of the SRI'' used in \cite{KanatsoulisFSM:hsr}, in order to avoid confusion with identifiability of the low-rank model.}.
\begin{theorem}\label{thm:TuckerIdentifiabilityDeterministic}
Let a Tucker decomposition of $\tY$ be
\begin{equation}\label{eq:Tucker_model}
\tY = \mlprod{\tG}{{\mU}}{{\mV}}{{\mW}},
\end{equation}
where  $\tens{G} \in \RR^{R_1 \times R_2 \times R_3}$, and $\mU \in \RR^{I \times R_1}$, $\mV \in \RR^{J \times R_2}$,  $\mW \in \RR^{K \times R_3}$   have full column rank. 
We also assume that $ \tens{E}_M, \tens{E}_H = \matr{0}$ in \eqref{eq:lmm}.
\begin{enumerate}
\item If 
\begin{equation}\label{eq:unfolding_rank_persistence}
\begin{split}
&\rank{\unfold{Y}{1}_M} = R_1,  \rank{\unfold{Y}{2}_M} = R_2,  \\
&\rank{\unfold{Y}{3}_H} = R_3,
\end{split}
\end{equation}
and one of the following conditions holds true:
\begin{enumerate}
\item[a)] either $\rank{\matr{P}_1 \matr{U}} =R_1$ and $\rank{\matr{P}_2 \matr{V}}=R_2$; 
\item[b)] or $\rank{\matr{P}_M \matr{W}}  = R_3$.
\end{enumerate}
Then there exists only one $\widehat{\tY}$  with multilinear rank at most $(R_1, R_2, R_3)$   such that $ \widehat{\tY} \con_3 {\matr{P}_M}  = {\tY}_M$ and $\widehat{\tY} \con_1 {\matr{P}_1} \con_2 {\matr{P}_2} = {\tY}_H$.
\item If  $\mU$, $\mV$, $\mW$ none of the conditions a) and b) are  satisfied, then 
 there exist infinitely many $\widehat{\tY}$  of the form  
 \[
\begin{split}
& \widehat{\tY} =  \mlprod{\widehat{\tens{G}}}{\widehat{\matr{U}}}{\widehat{\matr{V}}}{\widehat{\matr{W}}}, \\
 &\widehat{\matr{U}} \in \RR^{I \times R_1}, \widehat{\matr{V}} \in \RR^{J \times R_2}, \widehat{\matr{W}} \in \RR^{K \times R_3},
\end{split}
 \]
such that $ \widehat{\tY} \con_3 {\matr{P}_M} = {\tY}_M$ and $\widehat{\tY} \con_1 {\matr{P}_1} \con_2 {\matr{P}_2} = {\tY}_H$; 
in fact,  $\|\widehat{\tY} - \tY\|$ can be 
arbitrary large  for such $\widehat{\tY}$.  
\end{enumerate}
\end{theorem}
\begin{proof}
First of all, we note that by \cite[Theorem 13.16]{Laub04:siam}, the singular values of the matrix $\matr{X}^{\T} \matr{X} = \matr{I}_{R_3}\kron\mA + \matr{D}\kron\matr{I}_{R_1R_2}$ in \eqref{eq:left}  are all  sums of the pairs of eigenvalues of 
\begin{equation}\label{eq:krprod}
\underbrace{\left(\widehat{\mV}^{\T}\mP_2^{\T}\mP_2\widehat{\mV}\right)\kron\left(\widehat{\mU}^{\T}\mP_1^{\T}\mP_1\widehat{\mU}\right)}_{\mA}, \underbrace{\lambda\widehat{\mW}^{\T}\mP_M^{\T}\mP_M\widehat{\mW}}_{\matr{D}}.
\end{equation}
We also assume without loss of generality that  ${\matr{U}},{\matr{V}}, {\matr{W}}$ have orthonormal columns.

\begin{itemize}
\item\noindent\emph{Proof of 2)} Assume that $\rank{\matr{P}_1 \matr{U}}\rank{\matr{P}_2 \matr{U}} < R_1R_2$ and $\rank{\matr{P}_M \matr{W}}  < R_3$.
If we set $\widehat{\matr{U}} = {\matr{U}}$, $\widehat{\matr{V}} = {\matr{V}}$, $\widehat{\matr{W}} = {\matr{W}}$, then $\rank{\mA} < R_1R_2$, $\rank{\matr{D}}<R_3$ and $\rank{\matr{X}^{\T} \matr{X}} < R_1 R_2 R_3$.
Therefore the system \eqref{eq:ls_core} is underdetermined, and there is an infinite number of solutions $\widehat{\tens{G}} \in \RR^{R_1 \times R_2 \times R_3}$.
Note that if we define $\widehat{\tens{Y}} = \mlprod{\widehat{\tens{G}}}{\matr{U}}{\matr{V}}{\matr{W}}$, then it is an admissible solution, i.e.,
$ \widehat{\tY} \con_3 {\matr{P}_M} = {\tY}_M$ and $\widehat{\tY} \con_1 {\matr{P}_1} \con_2 {\matr{P}_2} = {\tY}_H$. On the other hand, due to orthogonality of the bases,  $\|\widehat{\tens{Y}} - {\tens{Y}}\|_F = \|\widehat{\tens{G}} - {\tens{G}}\|_F$, which can be made arbitrary large due to nonuniqueness of the solution to \eqref{eq:ls_core}.

\item\noindent\emph{Proof of 1)} Let us choose $\widehat{\matr{U}} \in \RR^{I \times R_1}$, $\widehat{\matr{V}} \in \RR^{J \times R_2}$, and $\widehat{\matr{W}} \in  \RR^{K \times R_3}$ to be orthogonal bases of the row spaces of ${\unfold{Y}{1}_M}$, ${\unfold{Y}{2}_M}$ and ${\unfold{Y}{3}_H}$ respectively.
First, by \eqref{eq:unfolding_rank_persistence}, the rank of unfoldings does not drop after degradation, hence 
\[
\widehat{\matr{U}}  = \matr{U} \matr{Q}_U,  \widehat{\matr{V}}  = \matr{V} \matr{Q}_V, \widehat{\matr{W}}  = \matr{W} \matr{Q}_W,
\]
where $\matr{Q}_U$, $\matr{Q}_V$, $\matr{Q}_W$ are some rotation matrices.
Next,  due to conditions on the ranks of ${\matr{P}_1 \matr{U}}$, ${\matr{P}_2 \matr{U}}$ and ${\matr{P}_M \matr{W}}$, we get that $\rank{\matr{X}^{\T} \matr{X}} = R_1 R_2 R_3$ because of (\ref{eq:krprod}). Hence the solution $\widehat{\tens{G}}$ of \eqref{eq:ls_core} is unique.
Finally, we note that the reconstructed tensor can be expressed as
\[
\vecl{\widehat{\tY}} = 
(\matr{\widehat{W}}\kron\matr{\widehat{V}}\kron\matr{\widehat{U}}) (\matr{X}^{\T} \matr{X})^{-1} \matr{X}^{\T}\vect{z},
\]
where the right-hand side does not depend on the rotation matrices $\matr{Q}_U$, $\matr{Q}_V$, and $\matr{Q}_W$ due to the definition of $\matr{X}$. Hence, the reconstructed tensor $\widehat{\tY}$ is unique. \qedhere
\end{itemize}
\end{proof}

\begin{corollary}
If  the conditions of  \Cref{thm:TuckerIdentifiabilityDeterministic} (part 1) hold, 
 then  any  minimizer  of  \eqref{eq:HSR_Tucker_cost} recovers $\tY$, i.e.
\[
{\tY} = \mlprod{\widehat{\tG}}{\widehat{\mU}}{\widehat{\mV}}{\widehat{\mW}}.
\]
In addition,  Algorithm~\ref{alg:hosvd}  recovers $\tY$  for all cases of recoverability in \Cref{thm:TuckerIdentifiabilityDeterministic}.
\end{corollary}

The recoverability results derived in \Cref{thm:TuckerIdentifiabilityDeterministic} are valid if a Tucker decomposition is used, and if its core tensor is dense. 
But they still remain valid if the core tensor is diagonal or block diagonal. 
For this reason, they also apply to CPD or BTD decompositions if the tensor rank is smaller than dimensions. 
In particular, recoverability can be ensured under mild conditions when the CPD is not unique, e.g. in the presence of collinear factors, as shown in the following corollary.

\begin{corollary}[Recoverability  for CPD model  with partial uniqueness]\label{cor:recoveryNonidentifiable}
Assume that the SRI  has a CPD ${\tY} = \cpd{\mA}{\mB}{\mC}$ of rank $F \le \min(I_H,J_H)$, such that
\[
\rank{\mA} = \rank{\mP_1 \mA} = \rank{\mB} = \rank{\mP_2 \mB} = F
\]
and $\mP_M \mC$ does not have zero columns.
We also assume that $ \tens{E}_M, \tens{E}_H = \matr{0}$ in \eqref{eq:lmm}.
Then any minimizer of \eqref{eq:coupled_cpd} recovers ${\tY}$.
\end{corollary}
\begin{proof}
Since the original factors $\mA,\mB,\mC$ yield zero error in \eqref{eq:coupled_cpd}, hence any global minimizer $(\widehat{\mA},\widehat{\mB},\widehat{\mC})$ of  \eqref{eq:coupled_cpd}, satisfies 
\[
\cpd{\mP_1\widehat{\mA}}{\mP_2\widehat{\mB}}{\widehat{\mC}} = \tY_H \mbox{ and } \cpd{\widehat{\mA}}{\widehat{\mB}}{\mP_M \widehat{\mC}} = \tY_M.
\]
Due to the conditions of the corollary, $\cpd{\mP_1{\mA}}{\mP_2{\mB}}{{\mC}}$ and $\cpd{{\mA}}{{\mB}}{\mP_M {\mC}}$ satisfy partial uniqueness conditions in \cite[Theorem 2.2]{GuoMBS12:tens}. 
Hence (after permutations and rescaling of factors), we have $\widehat{\mC} = \mC$ and
\[
\rank{\mC} = \rank{\widehat{\mC}} = \rank{\unfold{Y}{3}_H}= R_3.
\]
Moreover, since 
\[
\unfold{Y}{1}_M = (\mP_M \matr{C} \kr \matr{B}) \matr{A}^\T,\quad \unfold{Y}{2}_M = (\mP_M \matr{C} \kr \matr{A}) \matr{B}^\T,
\]
and $\mP_M \matr{C}$ does not have zero columns, we have that
\[
\begin{split}
&\rank{\mA} = \rank{\widehat{\mA}} = \rank{\unfold{Y}{1}_M}= R_1=  F,  \\
&\rank{\mB} = \rank{\widehat{\mB}} = \rank{\unfold{Y}{2}_M} = R_2 = F.  \\
\end{split}
\]
Therefore, both $(\mA,\mB,\mC)$ and $(\widehat{\mA},\widehat{\mB},\widehat{\mC})$ are particular solutions of Problem \eqref{eq:HSR_Tucker_cost} with an additional constraint that the tensor rank of  ${\tG}$ is at most $F$.
Since, by Theorem~\ref{thm:TuckerIdentifiabilityDeterministic}, any solution of \eqref{eq:HSR_Tucker_cost} recovers $\tY$ uniquely, the proof is complete.
\end{proof}
Note that the conditions of Corollary~\ref{cor:recoveryNonidentifiable} { are quite restrictive for real applications. But, they probably} can be relaxed by using Kruskal ranks and  a   more general formulation in \cite[Theorem 2.1]{GuoMBS12:tens} (see also \cite{DL08:tens}).

\subsection{Exact recoverability for generic tensors}
From the deterministic recovery conditions, we can establish the generic recoverability results.
\begin{theorem}\label{thm:TuckerIdentifiability}
Assume that $\matr{P}_1 \in \RR^{I_H \times I}$, $\matr{P}_2 \in \RR^{J_H \times J}$, and  $\matr{P}_M \in \RR^{K_M \times K}$ are fixed full  row-rank  matrices.
Let $\tY$ have decomposition \eqref{eq:Tucker_model},
where $R_1 \le I$, $R_2 \le J$, $R_3 \le K$, and $\tens{G} \in \RR^{R_1 \times R_2 \times R_3}$,  $\mU \in \RR^{I \times R_1}$, $\mV \in \RR^{J \times R_2}$,  $\mW\in \RR^{K \times R_3}$ are random tensors and matrices, distributed according to an absolutely  continuous  probability distribution.
We also assume that $ \tens{E}_M, \tens{E}_H = \matr{0}$ in \eqref{eq:lmm}.
 
\begin{enumerate}
\item If $R_3 \le K_M$ or $(R_1, R_2) \le (I_H,J_H)$ and
\begin{equation}\label{eq:additional_conditions}
\begin{cases}
R_1 \le \min(R_3,K_M)R_2, \\
 R_2 \le \min(R_3,K_M)R_1, \\
 R_3 \le   \min(R_1,I_H) \min(R_2, J_H ),
 \end{cases}
\end{equation}
 then  with probability $1$  there exists a unique tensor $\widehat{\tY}$   with multilinear rank at most $(R_1, R_2, R_3)$    such that $\widehat{\tY} \con_3 {\matr{P}_M}= {\tY}_M$ and $\widehat{\tY} \con_1 {\matr{P}_1} \con_2 {\matr{P}_2} = {\tY}_H$.
\item If $R_3 > K_M$ and ($R_1 \!>\! I_H$ or $R_2 \!>\! J_H$), then  with probability $1$  the reconstruction is non-unique, i.e. there exist infinitely many  $\widehat{\tY}$
 of the form  
 \[
\begin{split}
& \widehat{\tY} =  \mlprod{\widehat{\tens{G}}}{\widehat{\matr{U}}}{\widehat{\matr{V}}}{\widehat{\matr{W}}}, \\
 &\widehat{\matr{U}} \in \RR^{I \times R_1}, \widehat{\matr{V}} \in \RR^{J \times R_2}, \widehat{\matr{W}} \in \RR^{K \times R_3},
\end{split}
 \]
such that $ \widehat{\tY} \con_3 {\matr{P}_M} = {\tY}_M$ and $\widehat{\tY} \con_1\! {\matr{P}_1} \con_2 \!{\matr{P}_2} = {\tY}_H$; in fact, $\|\widehat{\tY} - \tY\|$ can be arbitrary large  for such $\widehat{\tY}$.  
\end{enumerate}
\end{theorem}
\begin{proof}
\begin{itemize}

\item\noindent\emph{Proof of 2)}  follows  from \Cref{thm:TuckerIdentifiabilityDeterministic} (part 2)

\item \noindent\emph{Proof of 1)} First,  without loss of generality, we can replace $\matr{P}_1$, $\matr{P}_2$, $\matr{P}_M$ with the following of same size:
\begin{equation}\label{eq:simple_degradation} 
\widetilde{\matr{P}}_1 = \bmx \mI_{I_H} \\ \matr{0}\emx^{\T}\!\!,\,
\widetilde{\matr{P}}_2 = \bmx \mI_{J_H} \\ \matr{0}\emx^{\T}\!\!,\,
\widetilde{\matr{P}}_M = \bmx \mI_{K_M} \\ \matr{0}\emx^{\T}\!\!.
\end{equation}
Indeed, let us explain why it is so, for example for $\matr{P}_1 \in \RR^{I_H \times I}$.
There exists a nonsingular matrix\footnote{For example, $\matr{T} =  [\matr{P}_1^{\dag}\, \matr{F}],\mbox{ where }\matr{F} \in  \RR^{I \times (I - I_H)}, \matr{P}_1 \matr{F} = \matr{0}$.} $\matr{T}$ such that 
\[
\matr{P}_1  \matr{T} = \begin{bmatrix} \mI_{I_H} & \matr{0}\end{bmatrix} .
\]
If we take $\widetilde{\mU} = \matr{T}^{-1}\mU$  then $\matr{P}_1  \mU = \widetilde{\matr{P}}_1 \widetilde{\mU}$.
Note that a nonsingular transformation preserves absolute continuity of the distribution; hence $\mU$ has an absolutely continuous distribution if and only if $\widetilde{\mU}$ has one.

Therefore, under the assumptions on distribution of $\mU$, $\mV$, $\mW$ the following implications hold with probability~1
\[
\begin{split}
&R_1 \le I_H  \Rightarrow \rank{\matr{U}_{1:I_H,:}} =R_1, \\
&R_2 \le J_H\Rightarrow  \rank{ \matr{V}_{1:J_H,:}}=R_2, \\
&R_3 \le K_M \Rightarrow \rank{ \matr{W}_{1:K_M,:}}  = R_3.
\end{split}
\]
Next, we are going to show how the other set of conditions imply \eqref{eq:unfolding_rank_persistence}.
We will prove it only for the first condition (the others are analogous).

Note that  the first unfolding can be written as
\[
\unfold{Y}{1}_M= (\matr{W}_{1:K_M,:} \kron \matr{V}) \unfold{G}{1} \matr{U}^{\T}.
\]
Due to the dimensions of the terms in the product, this matrix is at most rank $R_1$.
Due to semicontinuity of the rank function, $\unfold{Y}{1}_H$ will be generically of rank $R_1$ if we can provide just a single example of  $\mU$, $\mV$, $\mW$, $\tG$, achieving the condition $\rank{\unfold{Y}{1}_M} = R_1$. 
Indeed, if $R_1 \le \min(R_3,K_M)R_2$, such an example is given by
\[
{\matr{U}} = \bmx \mI_{R_1} \\ \matr{0}\emx,
{\matr{V}} = \bmx \mI_{R_2} \\ \matr{0}\emx, 
{\matr{W}} = \bmx \mI_{R_3}\\  \matr{0}\emx,
\unfold{G}{1} = \bmx \mI_{R_1} \\ \matr{0}\emx,
\]
which completes the proof.
\qedhere
\end{itemize}

\end{proof}

We illustrate the statement of \Cref{thm:TuckerIdentifiability} for the case $I = J$, $I_H = J_H$  and $R_1 = R_2$. 
In \Cref{fig:identifiability_region} we show that the space of parameters $(R_1, R_3)$ is split into two regions: recoverable and non-recoverable.  The hatched area corresponds to the parameters where  condition \eqref{eq:additional_conditions} is not satisfied.

\begin{figure}[htb!]
  \centering
  \centerline{\includegraphics[width=7cm]{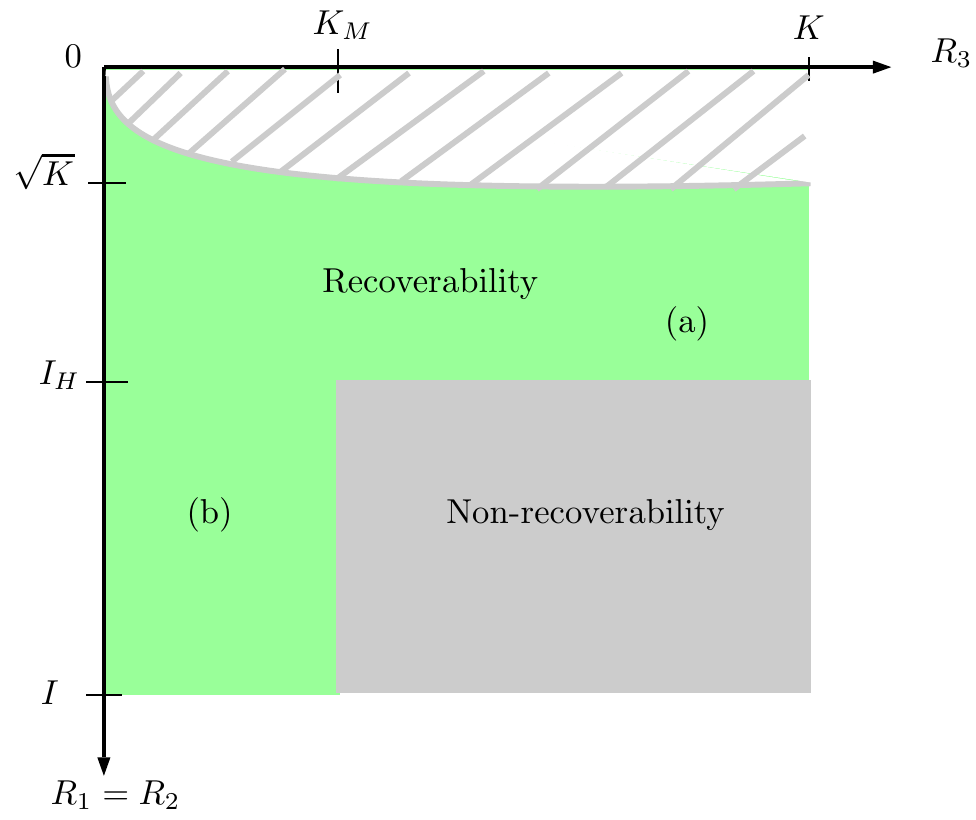}}

\caption{Identifiability region depending on $R_1$ and $R_3$}
\label{fig:identifiability_region}
\end{figure}

\begin{remark}
In the proof of  \Cref{thm:TuckerIdentifiability} it was shown that we can assume that the degradation operators are given in \eqref{eq:simple_degradation}.
In that case, the degraded tensors  ${\tY}_{M}$ and ${\tY}_{H}$ are just the subtensors (slabs)  \emph{i.e.}
\[
{\tY}_{M} = \tY_{:,:,1:K_M}, \quad {\tY}_{H} = \tY_{1:I_H,1:J_H,:}.
\]
Hence the  recoverability of Tucker super-resolution model is equivalent to uniqueness of tensor completion  \cite{Kres13:completion},   that is the recovery of $\tY$ from known subtensors $\tY_{:,:,1:K_M}$  and  $\tY_{1:I_H,1:J_H,:}$,  shown in Figure~\ref{fig:tensor_completion}.
\end{remark}

\begin{figure}[htb!]
  \centering
  \centerline{\includegraphics[width=7cm]{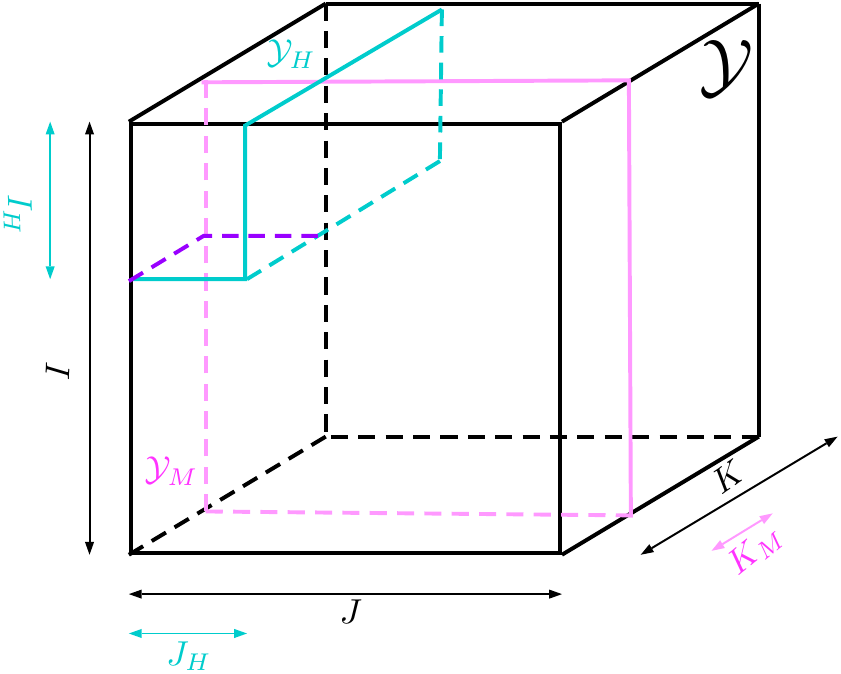}}

\caption{Recovery of $\tY$  from $\tY_{:,:,1:K_M}$ (pink), $\tY_{1:I_H,1:J_H,:}$ (blue)}
\label{fig:tensor_completion}
\end{figure}

We note that for  the coupled CP approach (STEREO), the connection with tensor completion  was  recently used in \cite{Kana19:fmri} for accelerated reconstruction of fMRI images.

\subsection{Recoverability in the blind case}
Similarly to \Cref{thm:TuckerIdentifiabilityDeterministic}, we can prove correct recovery for Algorithm~\ref{alg:hosvd_blind} under relaxed degradation model. 
We assume that the MSI is degraded as before, and HSI is degraded slicewise by an unknown linear operator $\mathcal{P}_s: \RR^{I\times J} \to \RR^{I_H\times J_H}$.
\begin{equation}\label{eq:degradation_blind}
\begin{cases}
{\tY}_M &= {\tY} \con_3 {\matr{P}_M}, \\
({\tY}_H)_{:,:,k} &= \mathcal{P}_s ({\tY}_{:,:,k}).
\end{cases}
\end{equation}
Then it is easy to prove the following analogue of \Cref{thm:TuckerIdentifiabilityDeterministic}.
\begin{proposition}
Let  $\tY$ have a Tucker decomposition 
\begin{equation*}
\tY = \mlprod{\tG}{{\mU}}{{\mV}}{{\mW}},
\end{equation*}
where  $\tens{G} \in \RR^{R_1 \times R_2 \times R_3}$, and $\mU \in \RR^{I \times R_1}$, $\mV \in \RR^{J \times R_2}$,  $\mW \in \RR^{K \times R_3}$ are full column rank.

If $\rank{\unfold{Y}{3}_H} = R_3$  and $\rank{\matr{P}_M \matr{W}}  = R_3$,
then Algorithm~\ref{alg:hosvd_blind} recovers  $\tY$ correctly. 
\end{proposition}
\begin{proof}
Indeed,  $\tY_M = \mlprod{\tG}{{\mU}}{{\mV}}{ {\matr{P}_M}{\mW}}$.
Therefore, since $\rank{{\matr{P}_M}{\mW}} = R_3$, the multilinear rank of $\tY_M$ is equal to the one of $\tY$  and
\[
\tY =  \tY_M \con_3 (\mW ({\matr{P}_M}{\mW})^{\dagger}).
\]
Finally, due to the condition $\rank{\unfold{Y}{3}_H} = R_3$, step 2 of Algorithm~\ref{alg:hosvd_blind} recovers $\mW$ up to a change of basis, i.e., $\matr{Z} = \mW \matr{O}$, where $\matr{O} \in \RR^{R_3 \times R_3}$ is an orthogonal matrix.
Finally, due to the properties of the pseudoinverse
\[
(\mW \matr{O}({\matr{P}_M}{\mW}\matr{O})^{\dagger}) = \mW ({\matr{P}_M}{\mW})^{\dagger},
\]
which completes the proof.
\end{proof}

\section{NUMERICAL EXPERIMENTS}
\label{sec:exp}
All simulations were run on a MacBook Pro with 2.3 GHz Intel Core i5 and 16GB RAM.
The code is implemented in MATLAB.
For basic tensor operations we used TensorLab  3.0 \cite{tensorlab}.  
The results are reproducible and the codes are available online at 
\url{https://github.com/cprevost4/HSR_Software}.

\subsection{Degradation model}

Experiments are conducted on a set of semi-real and synthetic examples, in which the groundtruth SRI is artificially degraded to $\tY_H$ and $\tY_M$ by the degradation matrices $\matr{P}_1$, $\matr{P}_2$ and $\matr{P}_M$ according to model~(\ref{eq:lmm}).

For spatial degradation, we follow the commonly used Wald's protocol \cite{WaldRM97:hsr}.  The matrices $\matr{P}_1$, $\matr{P}_2$ are computed with a separable Gaussian blurring kernel of size $q=9$. Then, downsampling is performed along each spatial dimension with a ratio $d=4$ between $I,J$ and $I_H,J_H$, as in \cite{KanatsoulisFSM:hsr}.
We refer to Appendix~\ref{app:2} for more details on the construction of $\matr{P}_1$, $\matr{P}_2$.

In this paper, we consider two spectral responses used to generate the spectral degradation matrix $\matr{P}_M$. 
In all the semi-real examples,  available online at \cite{Landsat18:data}, the bands corresponding to water absorption are first removed as in \cite{KanatsoulisFSM:hsr}.
The LANDSAT sensor spans the spectrum from 400nm to 2500nm for the HSI and produces a 6-band MSI corresponding to wavelengths 450--520nm (blue), 520--600nm (green), 630--690nm (red), 760--900nm (near-IR), 1550-1750nm (shortwave-IR) and 2050--2350nm (shortwave-IR2).
The second response corresponds to a QuickBird sensor, which spans the spectrum from 430nm to 860nm for the HSI and produces a 4-band MSI which bands correspond to wavelengths 430--545nm (blue), 466--620nm (green), 590--710nm (red) and 715--918nm (near-IR).
The spectral degradation matrix $\matr{P}_M$ is a selection-averaging matrix that selects the common spectral bands of the SRI and MSI.

\subsection{Metrics}
As for the experimental setup, we follow \cite{KanatsoulisFSM:hsr}; we compare the groundtruth SRI with the recovered SRI obtained by the algorithms.
The main performance metric used in comparisons is \emph{reconstruction Signal-to-Noise ratio} (R-SNR) used in \cite{Aiazzi11:hsr}:
\begin{equation} 
\text{R-SNR} = 10{\log}_{10}\left(\frac{ \|\tY\|_F^2}{\|\widehat{\tY} - \tY\|_F^2}\right).
\end{equation}
In addition to R-SNR, we consider different metrics from \cite{Aiazzi11:hsr} described below:
\begin{equation}
\text{CC } = \frac{1}{IJK} \left( \sum\limits_{k=1}^{K} \rho \left({{\tY}}_{:,:,k}, {\tens{\widehat{Y}}}_{:,:,k}  \right)   \right),
\end{equation}
where $\rho (\cdot , \cdot)$ is the Pearson correlation coefficient between the estimated and original spectral slices;
\begin{equation}
\text{SAM } = \frac{180}{\pi}\frac{1}{IJ} \sum\limits_{n=1}^{IJ}  \arccos\left(\frac{{\unfold{Y}{3}_{n,:}}^{\T}{\unfold{\widehat{Y}}{3}_{n,:}}}{\| {\unfold{Y}{3}_{n,:}} \|_2 \| {\unfold{\widehat{Y}}{3}_{n,:}} \|_2}     \right),
\end{equation}
which computes the angle between original and estimated fibers;
\begin{equation}
\text{ERGAS } = \frac{100}{d} \sqrt{\frac{1}{IJK} \sum\limits_{k=1}^{K} \frac{\| {{\tens{\widehat{Y}}}_{:,:,k} - {{\tY}}_{:,:,k}\|_F^2}}{{\mu}_k^2}      },
\end{equation}
where ${\mu}_k^2$ is the mean value of ${\tens{\hat{Y}}}_{:,:,k}$. ERGAS represents the relative dimensionless global error between the SRI and the estimate, which is the root mean-square error averaged by the size of the SRI.
We also show the computational time for each algorithm, given by the \verb?tic? and \verb?toc? functions of MATLAB.

\subsection{Semi-real data: comparison with other methods}

In this subsection, we showcase the capabilities of SCOTT and B-SCOTT and compare them with state-of-the-art methods. 

	\subsubsection{Non-blind algorithms}	
	We compare the performance of non-blind algorithms (\emph{i.e} STEREO and its initialization algorithm TenRec, and SCOTT).
	 We test various ranks for both algorithms.
	 For STEREO and TenRec, we use the implementation\footnote{TenRec may  be also made faster, e.g.,  by using  algebraic algorithms for CP approximation \cite{SanchezK}, but we did not optimize its speed in this paper.} of \cite{KanatsoulisFSM:hsr}, available online at \cite{Kanatsoulis:code}.
	 In other subsections, we use our implementation with fast solvers for the Sylvester equations (see Appendix~\ref{app:1}).
	  For HySure \cite{SimoesBAC15:hsr}, the groundtruth number of materials $E$  is chosen as the number of endmembers as in \cite{KanatsoulisFSM:hsr}. 
	This algorithm is applied in a non-blind fashion, meaning that the spatial\footnote{In fact, HySure has a different, convolutional degradation model, that is not necessarily separable.} and spectral degradation operators are not estimated but obtained from $\matr{P}_1$, $\matr{P}_2$ and $\matr{P}_M$. 
	The same model is applied to the FUSE algorithm \cite{WeiDT15:hsr}. 
	As a comparison, we also show the performance of B-SCOTT when no splitting is performed.
	
	The first dataset we consider is Indian Pines, where ${\tY} \in \RR^{144\times 144\times 200}$ is degraded by a LANDSAT sensor for the MSI and a downsampling ratio $d=4$ for the HSI.
	The results are presented in Tables~\ref{tab:IP_nb} and~\ref{tab:IP_nb_noisy}, and Figure~\ref{fig:visual_ip}. 
	In Table~\ref{tab:IP_nb} and following, the numbers between brackets represent the multilinear rank used for the Tucker approach. 
	

	\begin{table}[ht!]	
\centering	
{\small  
\begin{filecontents*}{exp3_table2.txt}
STEREO 50	2.689053e+01	8.845668e-01	2.258652e+00	1.035925e+00	2.173367e+00
STEREO 100	2.845695e+01	9.068142e-01	2.029706e+00	8.940159e-01	3.308821e+00
SCOTT [40,40,6]	2.627819e+01	8.837297e-01	2.358709e+00	1.076723e+00	2.948586e-01
SCOTT [30,30,16]	2.502546e+01	8.678535e-01	2.534930e+00	1.204603e+00	4.061686e-01
SCOTT [70,70,6]	2.633885e+01	8.832985e-01	2.518458e+00	1.129817e+00	3.493076e-01
SCOTT [24,24,25]	2.506170e+01	8.752479e-01	2.445984e+00	1.178489e+00	1.627471e-01
B-SCOTT [40,40,6]	2.512049e+01	8.687993e-01	2.758994e+00	1.254637e+00	1.071803e-01
HySure E = 16	2.643171e+01	8.713767e-01	2.496367e+00	1.103366e+00	1.517781e+01
FUSE	2.194549e+01	8.094603e-01	3.906815e+00	1.905482e+00	3.331815e-01
TenRec F = 50	2.681510e+01	8.834003e-01	2.270045e+00	1.048030e+00	8.429997e-01
TenRec F = 100	2.833757e+01	9.033765e-01	2.051300e+00	9.045725e-01	2.103804e+00
\end{filecontents*}
\pgfplotstabletypeset[header=false,
col sep = tab,
	columns/0/.style={string type},
every head row/.style={after row=\hline},
columns/0/.style={string type, column name={Algorithm}, column type/.add={@{}@{\,}}{}},
columns/1/.style={column name={R-SNR}, column type/.add={@{\,}|@{\,}}{}},
columns/2/.style={column name={CC}, column type/.add={@{\,}|@{\,}}{}},
columns/3/.style={column name={SAM}, column type/.add={@{\,}|@{\,}}{}},
columns/4/.style={column name={ERGAS}, column type/.add={@{\,}|@{\,}}{}},
columns/5/.style={column name={time}, fixed, precision=2, column type/.add={@{\,}|@{\,}}{}},
columns ={0,1,2,3,4,5},
]{exp3_table2.txt}}
\caption{Indian Pines (non-blind algorithms), no noise}
\label{tab:IP_nb}
\end{table}

	\begin{table}[ht!]	
\centering	
{\small  
\begin{filecontents*}{tab_reviews_ip.txt}
STEREO F = 50	2.535588e+01	8.544890e-01	2.621188e+00	1.201818e+00	2.111502e+00
STEREO F = 100	2.401146e+01	8.183063e-01	3.162457e+00	1.394010e+00	3.274191e+00
SCOTT [40,40,6]	2.367390e+01	8.348208e-01	3.138609e+00	1.459036e+00	3.258463e-01
SCOTT [30,30,16]	2.284741e+01	8.353951e-01	3.467894e+00	1.414638e+00	4.288622e-01
SCOTT [70,70,6]	1.799321e+01	7.669852e-01	5.919248e+00	2.881040e+00	3.571648e-01
SCOTT [24,24,25]	2.368887e+01	8.484413e-01	3.083206e+00	1.322740e+00	1.747316e-01
B-SCOTT [40,40,6]	2.025944e+01	8.055370e-01	4.654160e+00	2.197444e+00	1.149927e-01
HySure E = 16	2.043577e+01	7.237067e-01	4.834422e+00	2.086949e+00	1.570952e+01
FUSE 	1.220628e+01	6.203295e-01	1.116611e+01	5.655468e+00	3.518559e-01
TenRec F = 50	2.523861e+01	8.461542e-01	2.647268e+00	1.226699e+00	8.301614e-01
TenRec F = 100	2.428736e+01	8.166709e-01	3.008673e+00	1.365919e+00	2.087665e+00
\end{filecontents*}
\pgfplotstabletypeset[header=false,
col sep = tab,
	columns/0/.style={string type},
every head row/.style={after row=\hline},
columns/0/.style={string type, column name={Algorithm}, column type/.add={@{}@{\,}}{}},
columns/1/.style={column name={R-SNR}, column type/.add={@{\,}|@{\,}}{}},
columns/2/.style={column name={CC}, column type/.add={@{\,}|@{\,}}{}},
columns/3/.style={column name={SAM}, column type/.add={@{\,}|@{\,}}{}},
columns/4/.style={column name={ERGAS}, column type/.add={@{\,}|@{\,}}{}},
columns/5/.style={column name={time}, fixed, precision=2, column type/.add={@{\,}|@{\,}}{}},
columns ={0,1,2,3,4,5},
]{tab_reviews_ip.txt}}
\caption{Indian Pines (non-blind algorithms) with noise}
\label{tab:IP_nb_noisy}
\end{table}


In the noiseless case (see Table~\ref{tab:IP_nb}), we can see that for multilinear ranks chosen in the recoverability region (see Figure~\ref{fig:identifiability_region}), SCOTT yields similar performance to the one of STEREO with  lower computation time. 
Moreover, contrary to \cite{KanatsoulisFSM:hsr} (where $F=50$ is taken for STEREO), we found out that  tensor rank $F=100$ yields better performance.

In Table~\ref{tab:IP_nb_noisy}, white Gaussian noise is added to $\tY_H$ and $\tY_M$ with an input SNR of 25dB.
In this case, as in \cite{KanatsoulisFSM:hsr}, tensor rank $F=50$ yields better performance.
For $F=100$, TenRec gives slightly better performance than STEREO. 
Compared with the noiseless case, the performance of STEREO and TenRec deteriorate slightly, while we observe a bigger loss of performance for other methods, including SCOTT and B-SCOTT. 

\begin{figure}[htb]
 \centering
 \centerline{\includegraphics[width=7cm]{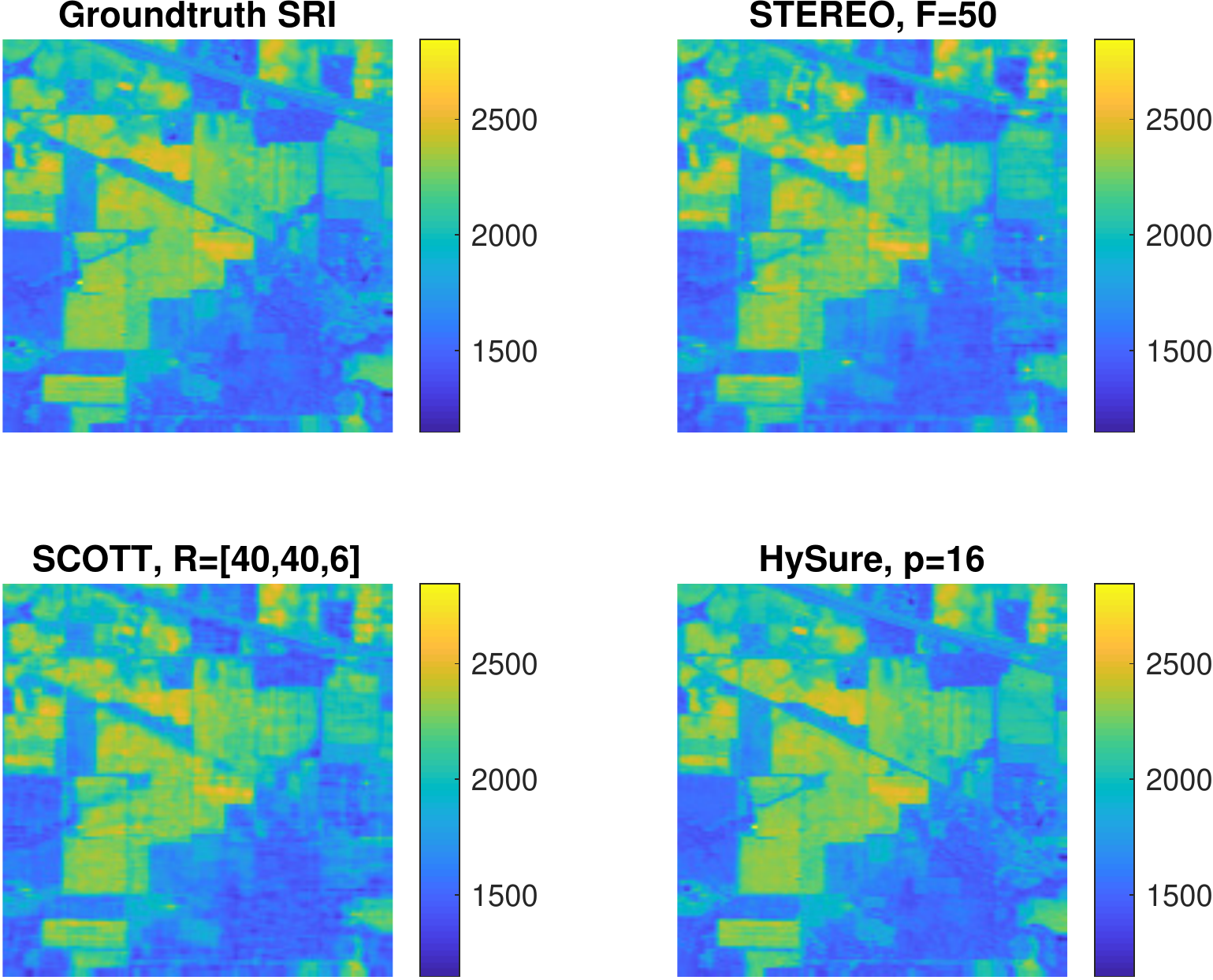}}
\caption{Spectral slice 120 of the SRI, Indian Pines}
\label{fig:visual_ip}
\end{figure}

	The other dataset is the Salinas-A scene, where  ${\tY} \in \RR^{80\times 84\times 204}$ is degraded with QuickBird specifications and $d=4$ for the HSI. 
	The results are presented in Table~\ref{tab:Sal_nb} and Figure~\ref{fig:visual_sal}.
	
		\begin{table}[ht!]
		\centering	
{\small 
\begin{filecontents*}{exp3_table2_sal.txt}
STEREO F = 50	3.340451e+01	9.705149e-01	9.114452e-01	3.347040e+00	9.188270e-01
STEREO F = 100	3.290921e+01	9.427868e-01	5.794443e-01	5.413065e+00	9.018093e-01
SCOTT [40,40,6]	3.151837e+01	9.504629e-01	7.076818e-01	4.920419e+00	9.924461e-02
SCOTT [50,50,6]	3.231273e+01	9.519425e-01	5.863406e-01	4.893006e+00	1.340482e-01
SCOTT [70,70,6]	3.288552e+01	9.528528e-01	4.799275e-01	4.886451e+00	2.771682e-01
B-SCOTT [40,40,6]	3.130393e+01	9.483149e-01	7.438577e-01	4.964291e+00	3.301995e-02
HySure E = 6	3.159375e+01	9.502414e-01	6.505538e-01	4.963045e+00	4.722608e+00
FUSE	2.061415e+01	8.789169e-01	1.906488e+00	5.891971e+00	1.189643e-01
TenRec F = 50	3.323677e+01	9.704034e-01	8.877365e-01	2.966725e+00	6.030667e-01
TenRec F = 100	2.914608e+01	9.240022e-01	6.309586e-01	9.788100e+00	8.354933e-01

\end{filecontents*}
\pgfplotstabletypeset[header=false,
col sep = tab,
	columns/0/.style={string type},
every head row/.style={after row=\hline},
columns/0/.style={string type, column name={Algorithm}, column type/.add={@{}@{\,}}{}},
columns/1/.style={column name={R-SNR}, column type/.add={@{\,}|@{\,}}{}},
columns/2/.style={column name={CC}, column type/.add={@{\,}|@{\,}}{}},
columns/3/.style={column name={SAM}, column type/.add={@{\,}|@{\,}}{}},
columns/4/.style={column name={ERGAS}, column type/.add={@{\,}|@{\,}}{}},
columns/5/.style={column name={time}, fixed, precision=2, column type/.add={@{\,}|@{\,}}{}},
columns ={0,1,2,3,4,5},
]{exp3_table2_sal.txt}}
\caption{Salinas A-scene (non-blind algorithms)}
\label{tab:Sal_nb}
\end{table}
\begin{figure}[ht!]
 \centering
 \centerline{\includegraphics[width=7cm]{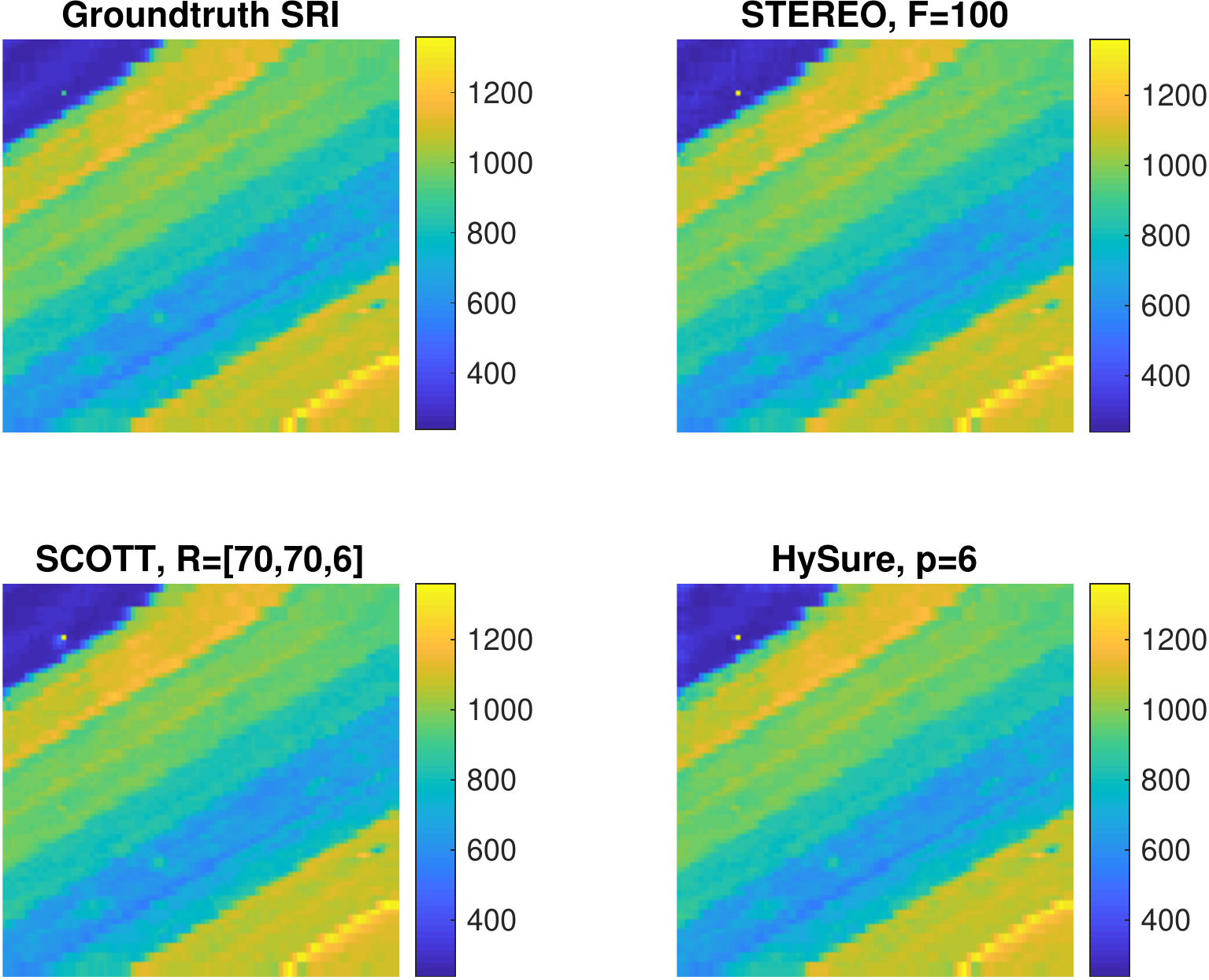}}
\caption{Spectral slice 120 of the SRI, Salinas-A scene}
 \label{fig:visual_sal}
\end{figure}

In \cite{KanatsoulisFSM:hsr}, CP-rank $F=100$ is used for STEREO.
However, we found out that for STEREO, tensor rank $F=50$ yields better reconstruction.
In Figure~\ref{fig:visual_sal}, we can see that STEREO and SCOTT can recover accurately the SRI.

	\subsubsection{Blind algorithms}
	
	We now consider the case where the spatial degradation matrices $\matr{P}_1$, $\matr{P}_2$ are unknown and compare the performance of B-SCOTT with Blind-STEREO \cite{KanatsoulisFSM:hsr}, SCUBA \cite{KanaFSM18:scuba} and HySure.
	We also compare the proposed blind algorithm to Blind-TenRec, the algebraic initialization of Blind-STEREO. 
	White gaussian noise is added to the HSI and MSI, with a SNR of 15dB and 25dB, respectively. 
	We consider two other datasets; the first one is a portion of the Pavia University, where $\tY \in \RR^{608\times 366\times 103}$ is degraded with QuickBird specifications for the MSI and $d=4$ for the HSI.
	We demonstrate the results in Table~\ref{tab:bPavia} and Figure~\ref{fig:visual_pavia} for visual reconstruction.
	For B-SCOTT, in the case where $R = [152,84,3]$, no compression is performed.
	In the following tables, the numbers between parentheses denote the number of blocks in which the HSI and MSI are split.
	For SCUBA, the numbers between brackets represent $[F, R_3]$. 
	
\begin{table}[ht!]	
\centering
{\small 
\begin{filecontents*}{tab4.txt}
SCUBA (4,4) [120,3]	2.566795e+01	9.905662e-01	3.243633e+00	1.965414e+00	1.897343e+01
B-SCOTT (4,4) [60,60,3]	2.340600e+01	9.858967e-01	3.831455e+00	2.373495e+00	4.392528e-01
B-SCOTT (4,4) [152,84,3]	2.641763e+01	9.917650e-01	2.972662e+00	1.830079e+00	5.745212e-01
B-SCOTT (4,4) [120,60,4]	2.562700e+01	9.910200e-01	3.014115e+00	1.811672e+00	4.609067e-01
SCUBA (8,8) [120,3]	2.649433e+01	9.918864e-01	2.925206e+00	1.811766e+00	5.015029e+01
B-SCOTT (8,8) [70,40,3]	2.651609e+01	9.919042e-01	2.921073e+00	1.813814e+00	5.670992e-01
B-STEREO F = 300	2.310854e+01	9.847269e-01	3.972021e+00	2.464133e+00	8.286244e+01
HySure E = 9	2.626864e+01	9.920292e-01	2.826676e+00	1.737321e+00	1.153863e+02
B-TenRec F = 300	2.251637e+01	9.828097e-01	4.056486e+00	2.648911e+00	2.951033e+01
\end{filecontents*}
\pgfplotstabletypeset[header=false,
col sep = tab,
	columns/0/.style={string type},
every head row/.style={after row=\hline},
columns/0/.style={string type, column name={Algorithm}, column type/.add={@{}@{\,}}{}},
columns/1/.style={column name={R-SNR}, column type/.add={@{\,}|@{\,}}{}},
columns/2/.style={column name={CC}, column type/.add={@{\,}|@{\,}}{}},
columns/3/.style={column name={SAM}, column type/.add={@{\,}|@{\,}}{}},
columns/4/.style={column name={ERGAS}, column type/.add={@{\,}|@{\,}}{}},
columns/5/.style={column name={time}, fixed, precision=2, column type/.add={@{\,}|@{\,}}{}},
columns ={0,1,2,3,4,5},
]{tab4.txt}}
\caption{Pavia University (blind algorithms)}
\label{tab:bPavia}
\end{table}

\begin{figure}[htb]
 \centering
 \centerline{\includegraphics[width=7cm]{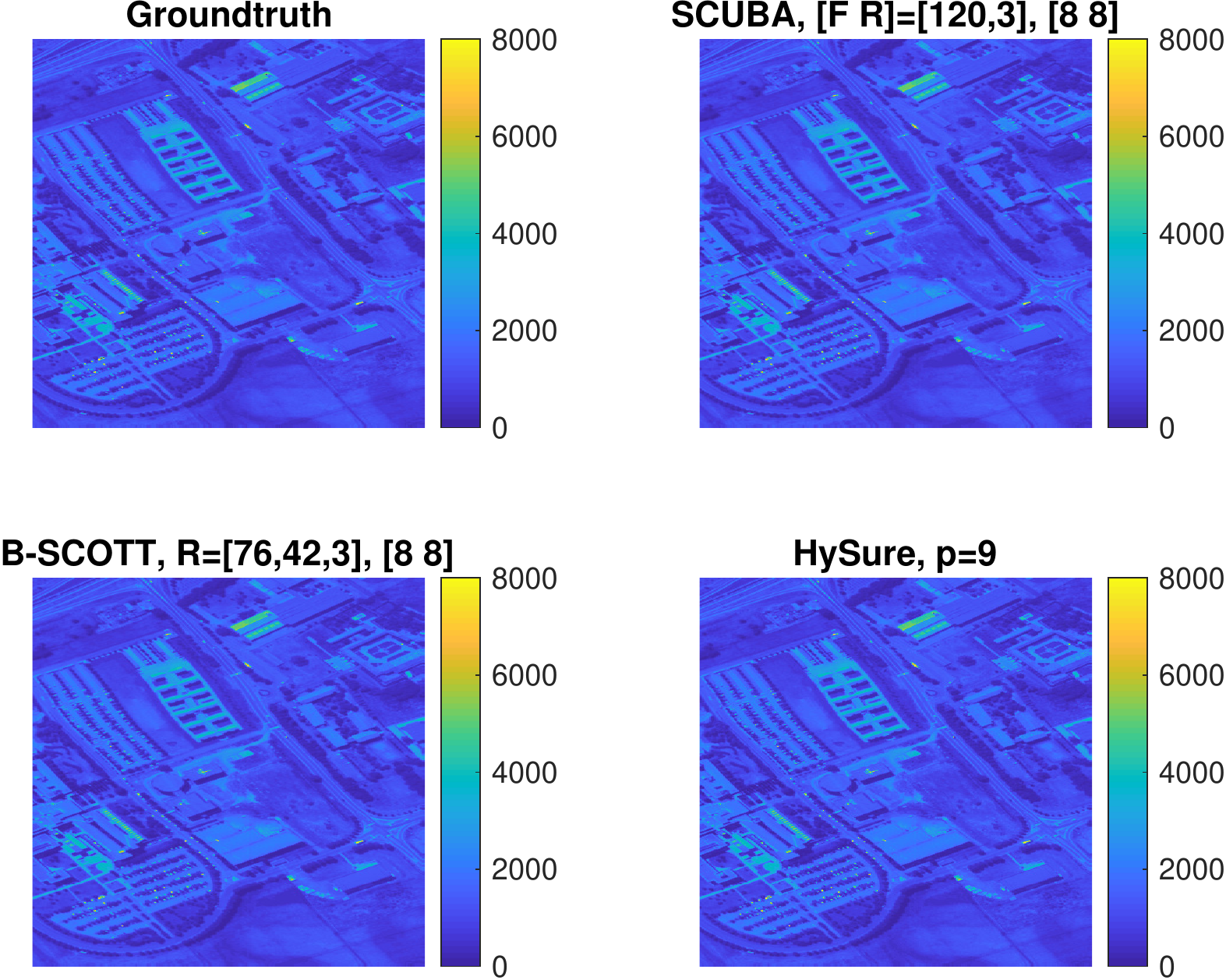}}
 \caption{Spectral slice 44 of the SRI, Pavia University}
\label{fig:visual_pavia}
\end{figure}
	
	\begin{table}[ht!]	
	\centering
{\small 
\begin{filecontents*}{tab5.txt}
SCUBA (4,4) [45,3]	3.170622e+01	9.731905e-01	1.119230e+00	6.567004e+00	1.257666e+01
B-SCOTT (4,4) [45,45,3]	3.191254e+01	9.743781e-01	1.083929e+00	6.573225e+00	8.910889e-01
B-SCOTT (4,4) [60,60,3]	3.302410e+01	9.794724e-01	1.029387e+00	6.575327e+00	1.193969e+00
SCUBA (8,8) [45,3]	3.465639e+01	9.853512e-01	9.156708e-01	6.165386e+00	3.379349e+01
B-SCOTT (8,8) [45,45,3]	3.469989e+01	9.855554e-01	9.054769e-01	6.191099e+00	1.251594e+00
B-STEREO F = 150	2.987457e+01	9.667193e-01	1.370901e+00	7.347774e+00	5.642040e+01
HySure E = 10	3.462318e+01	9.861489e-01	9.421252e-01	6.827986e+00	2.014047e+02
B-TenRec F = 150	3.070421e+01	9.676218e-01	1.211934e+00	6.476201e+00	1.266892e+01
\end{filecontents*}
\pgfplotstabletypeset[header=false,
col sep = tab,
	columns/0/.style={string type},
every head row/.style={after row=\hline},
columns/0/.style={string type, column name={Algorithm}, column type/.add={@{}@{\,}}{}},
columns/1/.style={column name={R-SNR}, column type/.add={@{\,}|@{\,}}{}},
columns/2/.style={column name={CC}, column type/.add={@{\,}|@{\,}}{}},
columns/3/.style={column name={SAM}, column type/.add={@{\,}|@{\,}}{}},
columns/4/.style={column name={ERGAS}, column type/.add={@{\,}|@{\,}}{}},
columns/5/.style={column name={time}, fixed, precision=2, column type/.add={@{\,}|@{\,}}{}},
columns ={0,1,2,3,4,5},
]{tab5.txt}}
\caption{ Cuprite (blind algorithms)}
\label{tab:bCuprite}
\end{table}

	In the second case, we consider the Cuprite dataset, where $\tY \in \RR^{512\times 614\times 224}$ is degraded with LANDSAT specifications and $d=4$.
	The results are presented in Table~\ref{tab:bCuprite} and Figure~\ref{fig:visual_cuprite}.

These two previous examples show that, for different splittings, and ranks taken from \cite{KanaFSM18:scuba}, B-SCOTT yields the best performance. 
For certain multilinear ranks, it even outperforms SCUBA with lower computation time.
Moreover, it outperforms Blind-STEREO and Blind-TenRec.
In terms of visual reconstruction, our algorithm can recover accurately the details of the groundtruth SRI, even though the spatial degradation matrices are unknown.

	\subsubsection{Hyperspectral pansharpening}
	
		\begin{figure}[htb]
 \centering
 \centerline{\includegraphics[width=9cm]{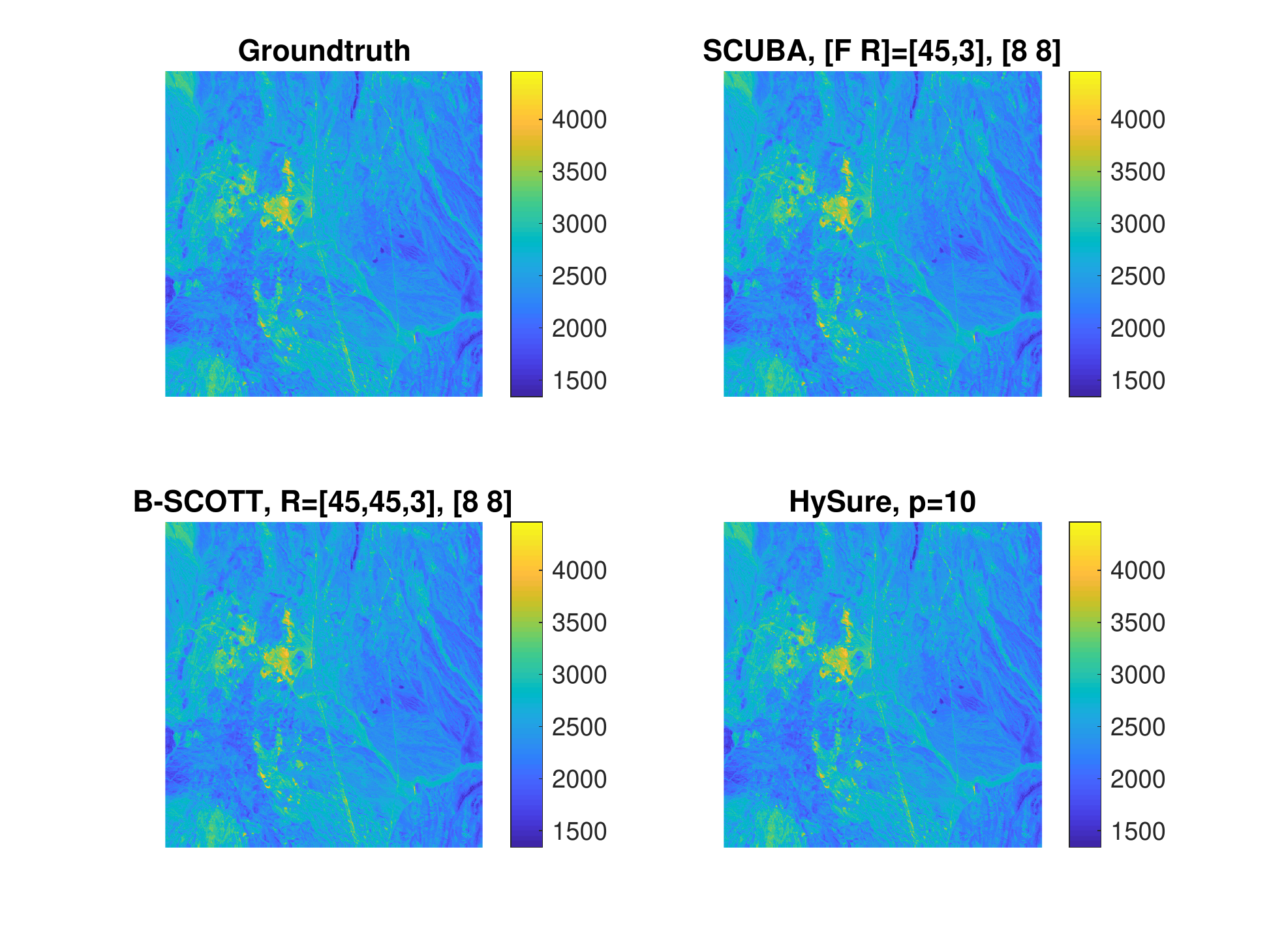}}
 \vspace{-0.9cm}
\caption{Spectral slice 44 of the SRI, Cuprite}
\label{fig:visual_cuprite}
\end{figure}

	Next, we address the pansharpening problem, which consists in fusion of a hyperspectral image and a panchromatic image (PAN) $\tY_P \in \RR^{I\times J\times 1}$. 
In this case, the spectral degradation matrix is obtained by averaging over the full spectral range of the groundtruth SRI, so that $\matr{P_M} \in \RR^{1\times K}$.
CP-based algorithms are not applicable, since their initialization is based on the CPD of the MSI (which is a matrix in the case of PAN images). 
In Table~\ref{tab:tab4res_ip}, the metrics are shown for different multilinear ranks for the Indian Pines dataset.
We also compare our results to those of HySure.
 We can see that even though the only possible value of $R_3$ is 1 for B-SCOTT, the algorithm still manages to yield a good recovery of the SRI.
On the other hand, SCOTT can also recover the SRI accurately, but is more sensitive to the choice of the multilinear rank.

	\begin{table}[ht!]	
	\centering
{\small 
\begin{filecontents*}{exp4_table2_ip.txt}
SCOTT [24,24,25]	2.047227e+01	7.747772e-01	4.407570e+00	1.953657e+00	2.186840e-01
SCOTT [30,30,16]	1.804322e+01	6.927842e-01	5.678894e+00	2.587458e+00	4.277833e-01
SCOTT [35,35,6]	1.460801e+01	5.430662e-01	7.837281e+00	3.885059e+00	8.871169e-01
HySure E = 16	2.066652e+01	7.550828e-01	4.237021e+00	1.987911e+00	1.439377e+01
BSCOTT (4,4) [24,24,1]	1.978357e+01	7.167046e-01	5.068709e+00	2.193548e+00	2.766742e-01
BSCOTT (4,4) [35,35,1]	1.978510e+01	7.167953e-01	5.068709e+00	2.193373e+00	1.205444e-01
\end{filecontents*}
\pgfplotstabletypeset[header=false,
col sep = tab,
	columns/0/.style={string type},
every head row/.style={after row=\hline},
columns/0/.style={string type, column name={Algorithm}, column type/.add={@{}@{\,}}{}},
columns/1/.style={column name={R-SNR}, column type/.add={@{\,}|@{\,}}{}},
columns/2/.style={column name={CC}, column type/.add={@{\,}|@{\,}}{}},
columns/3/.style={column name={SAM}, column type/.add={@{\,}|@{\,}}{}},
columns/4/.style={column name={ERGAS}, column type/.add={@{\,}|@{\,}}{}},
columns/5/.style={column name={time}, fixed, precision=2, column type/.add={@{\,}|@{\,}}{}},
columns ={0,1,2,3,4,5},
]{exp4_table2_ip.txt}}
\caption{Indian Pines (pansharpening)}
\label{tab:tab4res_ip}
\vspace{-0.9cm}
\end{table}
	
\subsection{Synthetic examples}

In most cases, generic recoverability conditions proposed in \cite{KanatsoulisFSM:hsr} are less restrictive than that of the Tucker approach.
The tensor rank $F$ can be larger than the dimensions of the SRI, while the multilinear ranks are bounded by its dimensions.
This gives the CP-based model better modelling power than the Tucker-based model, as shown for real data: regardless of the computation time, STEREO gives better performance than SCOTT.
However, there may exist deterministic cases in which the Tucker recoverability conditions are satisfied while nothing can be concluded from the results of \cite{KanatsoulisFSM:hsr}.
The goal of this subsection is to provide synthetic examples for such situations in the noiseless and noisy cases.
While these examples do not necesserily look like realistic hyperspectral images, they do help to better understand the recoverability conditions of the SRI and to evaluate their impact on the estimation performance.

	\subsubsection{Generating synthetic SRI}
	
	First, we explain how the synthetic SRI $\tY \in \RR^{I\times J\times K}$ are generated. 
	We consider $N$ spectral signatures $\vect{s}_1,\ldots, \vect{s}_N$ obtained from the Indian Pines groundtruth data \cite{Landsat18:data}.
	The SRI is split into $M^2$ equal blocks along the spatial dimensions.
	In each $\frac{I}{M}\!\times\!\frac{J}{M}$ block, at most one material is active, indicated by a number in the corresponding cell of a parcel map (see Table~\ref{tab:synth1} for an example).

    Formally, the SRI is computed as
     \begin{equation}\label{eq:synth}
     \tY = \sum\limits_{n=1}^{N} \matr{A}_n\otimes\vect{s}_n, 
     \end{equation}
        	\begin{table}[htb!]
	\centering
	\begin{tabular}{|c|c|}
      \hline
     1 & 2\\
      \hline
      2 & \\
      \hline
    \end{tabular}
    \vspace{0.2cm}
    \caption{Parcel map for $N=2$}
\label{tab:synth1}
    \end{table}
    
\noindent  where the abundance map $\matr{A}_n$ is a block matrix with Gaussians of fixed size present on the blocks corresponding to material $n$ in the parcel map.

	For instance, we consider the case presented in Table~\ref{tab:synth1}; the two abundance maps are
	\begin{equation*}
	\matr{A}_1 = \begin{bmatrix} \matr{H} & 0 \\ 0 & 0\end{bmatrix},\quad\matr{A}_2 = \begin{bmatrix} 0 & \matr{H} \\ \matr{H} & 0\end{bmatrix},
	\end{equation*}
	where $\matr{H}$ is a $60\times 60$ Gaussian with $\sigma = 20$.
	To illustrate this example, we show in Figure~\ref{fig:synth} two spectral bands of $\tY$.
	
	\begin{figure}[htb]
 \centering
 \vspace{-0.3cm}
 \centerline{\includegraphics[width=9cm]{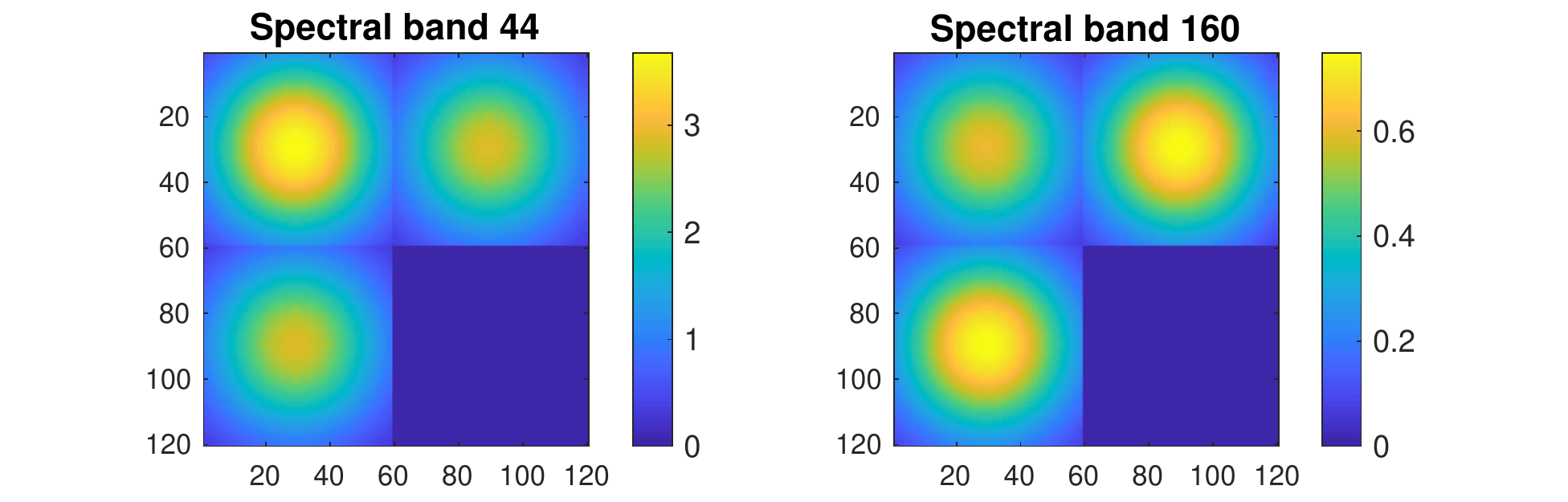}}
\caption{Spectral bands of the synthetic SRI with $N=2$}
\label{fig:synth}
\end{figure}

\subsubsection{Non-existing low-rank approximations}\label{sec:struct_examples}
	
	Let us consider the example introduced in Table~\ref{tab:synth1}.
	Due to separability of the Gaussians, $\tY$ has the following multilinear decomposition:
			\begin{align*}
	\tY = \mlprod{\tG}{{\mU}}{{\mV}}{{\matr{S}}},&\\
	\text{where } \tG_{:,:,1} =\begin{bmatrix} 1 & 0 \\ 0 & 0\end{bmatrix},\tG_{:,:,2} = \begin{bmatrix} 0 & 1 \\ 1 & 0\end{bmatrix},&\\ \mU = \mV = \begin{bmatrix} H & 0 \\ 0 & H \end{bmatrix}\text{ and } \matr{S} = [\vect{s}_1\quad \vect{s}_2]&.
	\end{align*}
	The multilinear rank of $\tY$ is equal to $(2,2,2)$, while the tensor rank of $\tY$ is equal to the tensor rank of $\tG$, which is known to be  equal to $F=3$ \cite[Ex.~2]{Comon09:tens}, \cite[Ex.~6.6]{ComonGLM08:tens}. 
	This is a well-known case where the best  rank-2 CP approximation does not exist \cite{BrachatCMT09:tens, Como02:oxford}, thus we can expect problems with the CP-based approach.

		We generate $\tY_H$ and $\tY_M$ with a downsampling ratio of $d=4$ for the HSI and LANDSAT specifications for the MSI (see Appendix~\ref{app:2}). 
		No noise is added to the MSI and HSI. 
		We run STEREO and TenRec  for $F$ in $[1:40]$ and SCOTT for $R_1=R_2$ in $[1:40]$  and $R_3$ in $[1:10]$ under recoverability conditions.
		For each algorithm, we compute the R-SNR as a function of the rank; the results are provided on Figure~\ref{fig:synth1}.
		As a comparison, on the same plot as STEREO and TenRec, we plot the results of SCOTT for $R_3 = N$ and $R_1 = R_2 = F$.
		
\begin{figure}[htb!]
\begin{minipage}[b]{0.49\linewidth}
 \centering
 {\includegraphics[height = 3.6cm]{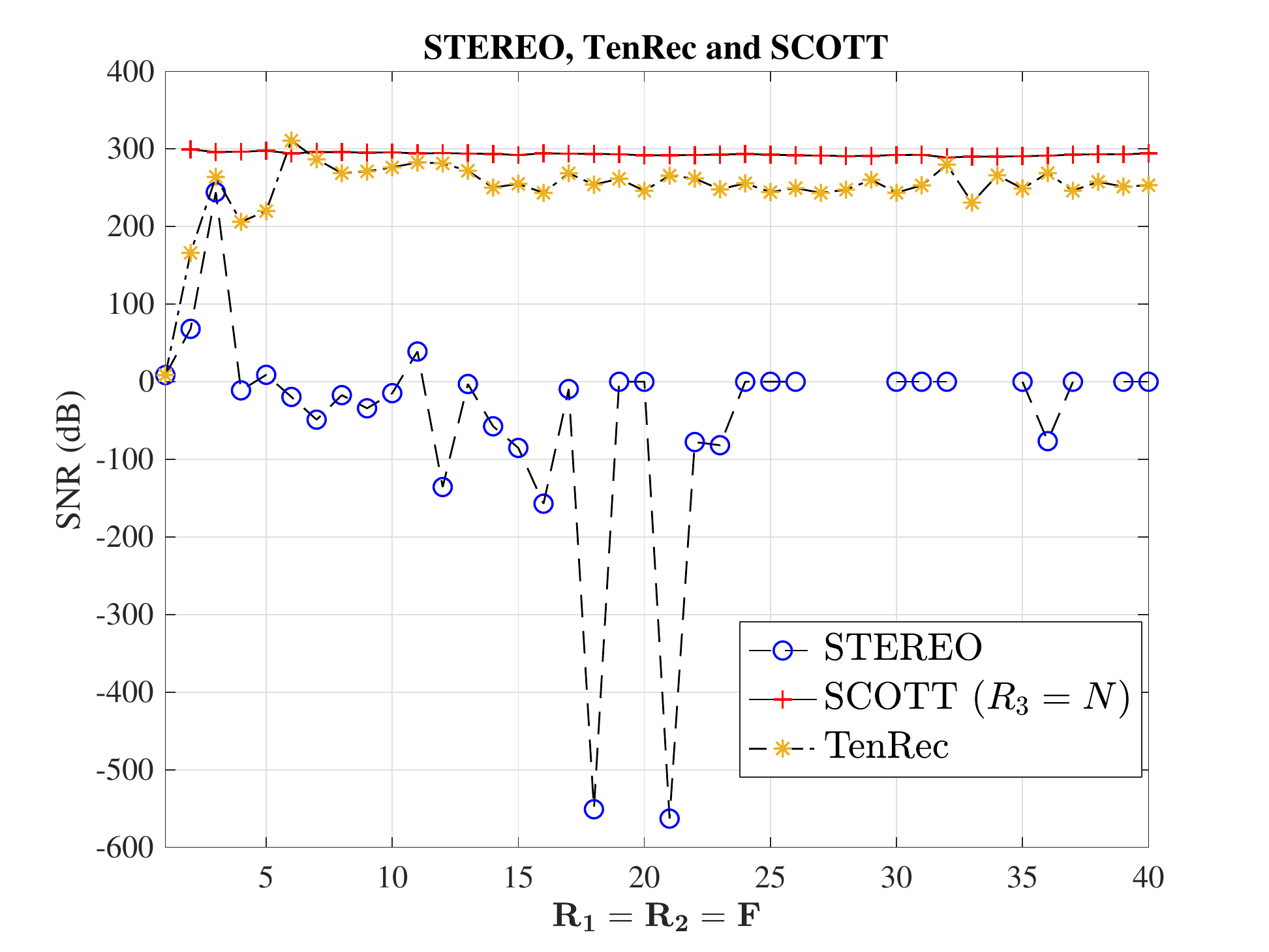}}
\end{minipage}
\hfill
\begin{minipage}[b]{0.49\linewidth}
 \centering
 {\includegraphics[height = 3.6cm]{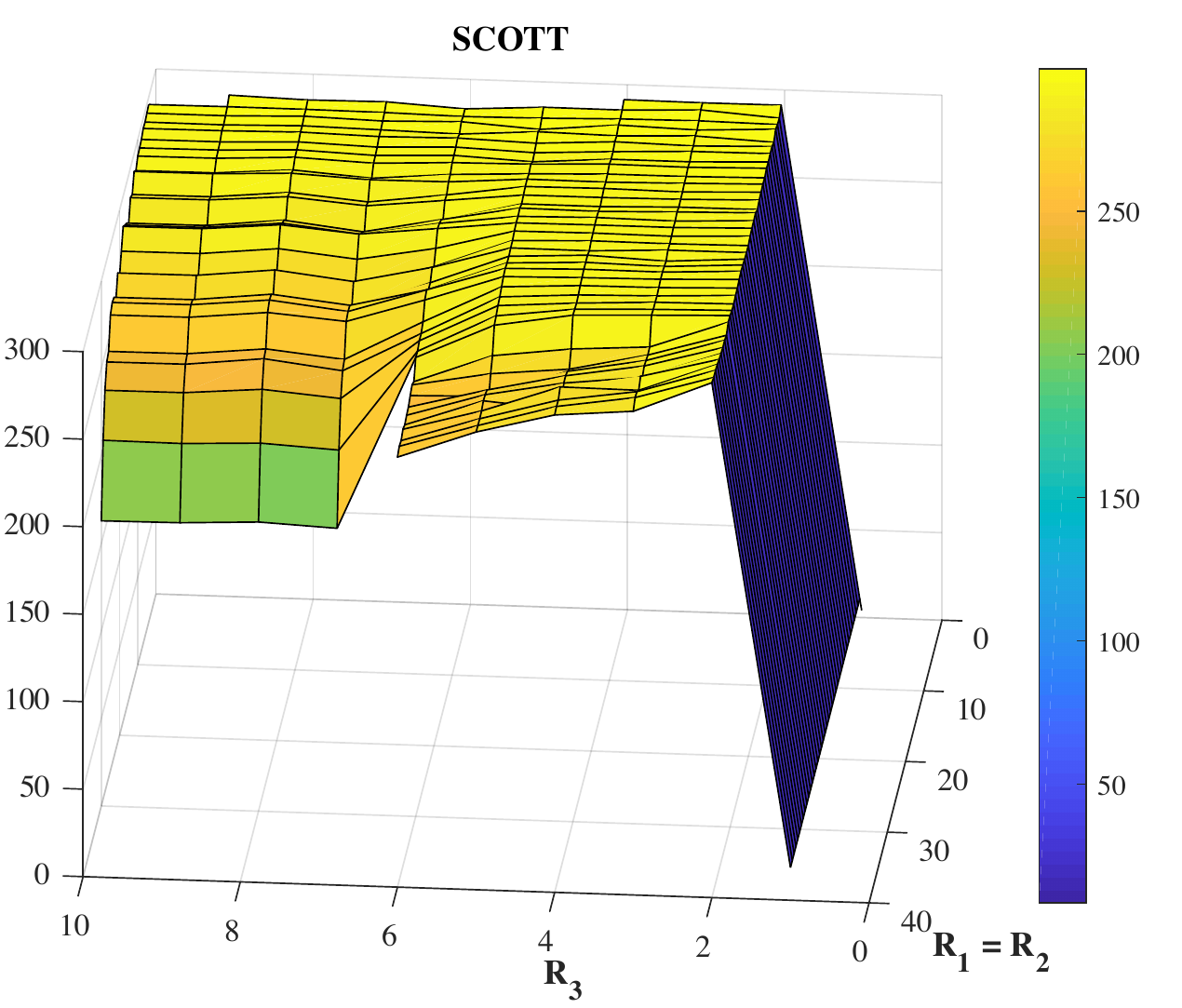}} 
\end{minipage}
\caption{R-SNR as a function of the rank}
\label{fig:synth1}
\end{figure}

For SCOTT, the best reconstruction error (given by R-SNR) is obtained for $R_3 = N$ and is rather insensitive to the choice of $R_1=R_2$.
$R_3$ can also be chosen larger than $N$ without significant loss of performance. 
For STEREO, only rank $F=3$ allows for an accurate reconstruction of the SRI.
For other tensor ranks, either the algorithm breaks (when no point is plotted, e.g. $F = 32$) or leads to inaccurate recovery.
TenRec however achieves the correct recovery for a wide range of tensor ranks. 
We can see that in this case, performing STEREO iterations after TenRec leads to a loss of performance.
We believe that this is due to the presence of colinear factors in the approximation causing ill-conditioning of ALS iterations. 
However, for noisy or real examples, this phenomenon is not likely to occur.

\subsubsection{Higher rank and noisy example}\label{sec:noisyex}

We also consider a slightly more realistic scenario.
The following example is made of $N=7$ materials, generated similarly to the previous example, as illustrated in Table~\ref{tab:ex2} and Figure~\ref{fig:synth5}.
The abundance maps are arranged along an anti-diagonal pattern, in a similar fashion as for the Salinas-A dataset.
    
    \begin{table}[htb!]
\centering
\begin{tabular}{|c|c|c|c|}
\hline
1 & 2 & 3 & 4 \\ \hline
2 & 3 & 4 & 5 \\ \hline
3 & 4 & 5 & 6 \\ \hline
4 & 5 & 6 & 7 \\ \hline
\end{tabular}
\vspace{0.2em}
        \caption{Parcel map for $N=7$}
        \label{tab:ex2}
\end{table}

    \begin{figure}[htb]
 \centering
 \vspace{-0.3cm}
 \centerline{\includegraphics[width=9cm]{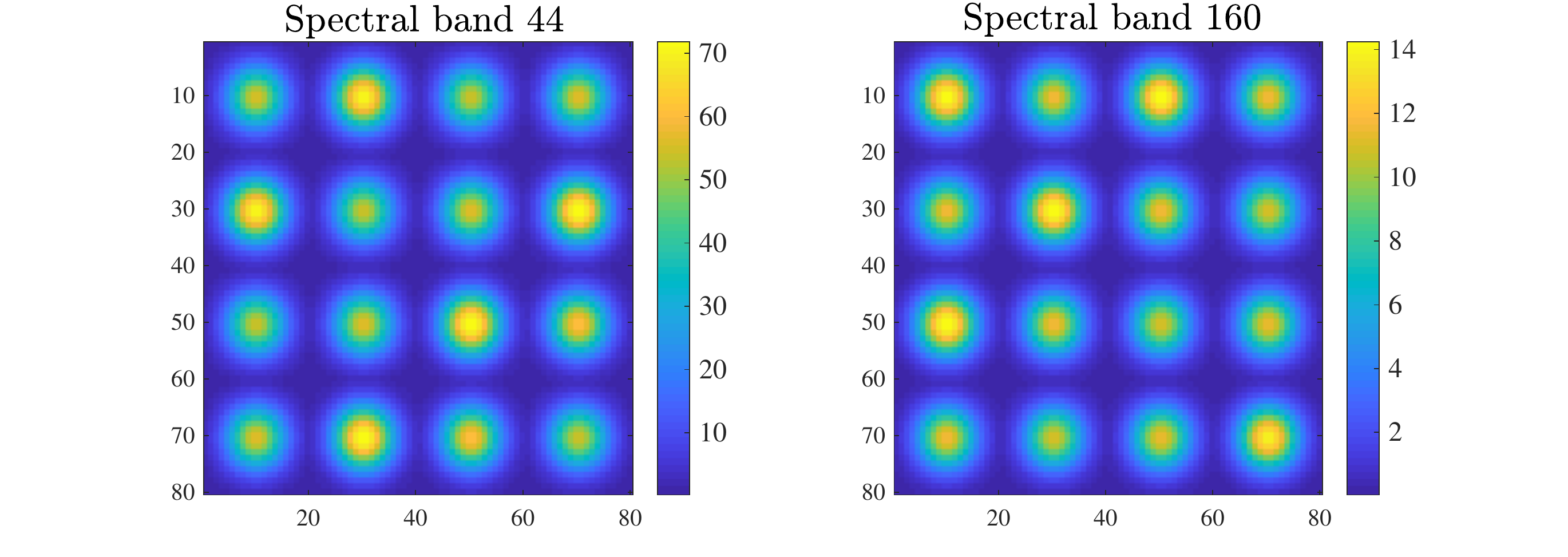}}
\caption{Spectral bands of the synthetic SRI with $N=7$}
\label{fig:synth5}
\end{figure}
    
In this example, $\tY \in \RR^{80\times 80\times 200}$ is degraded with QuickBird specifications for the MSI and $d=4$ for the HSI.
White Gaussian noise is added to the degraded images with an input SNR of 35dB.
The multilinear rank of the SRI is $\vect{R}=(4,4,7)$ while we do not know the tensor rank.
Similarly, we run both algorithms with the same setup as in the previous example, including a comparison of STEREO and TenRec, and SCOTT for $R_3 = N$ and an overestimated $R_3=15$. 
Results are presented in Figure~\ref{fig:synth2}.

For SCOTT, the best R-SNR is obtained for $\vect{R} = (4,4,7)$, which is the multilinear rank of the noiseless tensor.
Moreover, the best reconstruction error is obtained for $R_3 = N$: in this case, the performance of SCOTT is better than that of CP-based approaches.
SCOTT is also robust to an overestimation of $R_3$ or $R_1=R_2$.
TenRec and STEREO have almost the same performance, which is lower than that of SCOTT in this example.

 
 \begin{figure}[htb!]
\begin{minipage}[b]{0.49\linewidth}
 \centering
 {\includegraphics[height=3.5cm]{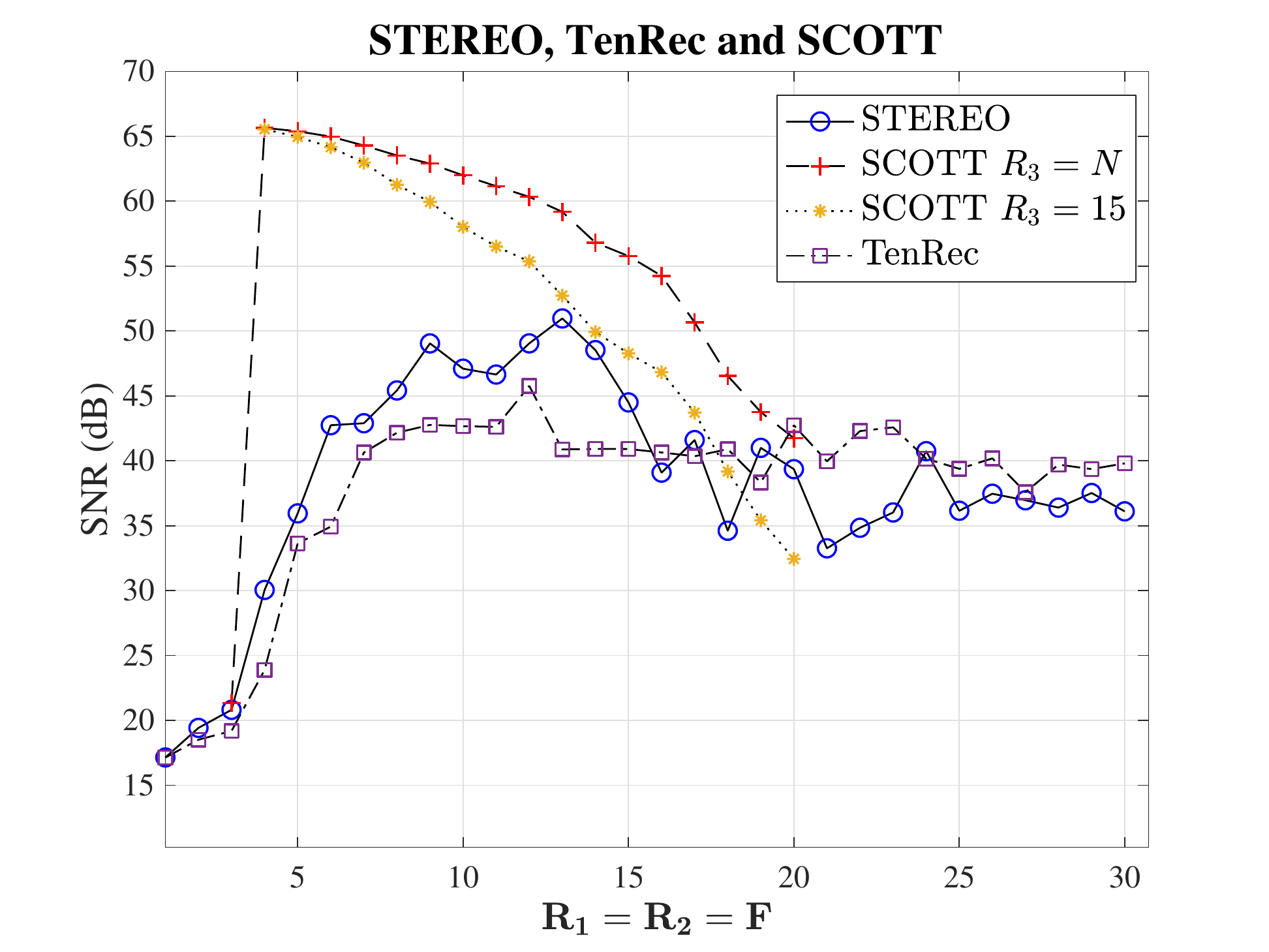}}
\end{minipage}
\hfill
\begin{minipage}[b]{0.49\linewidth}
 \centering
 {\includegraphics[height=3.5cm]{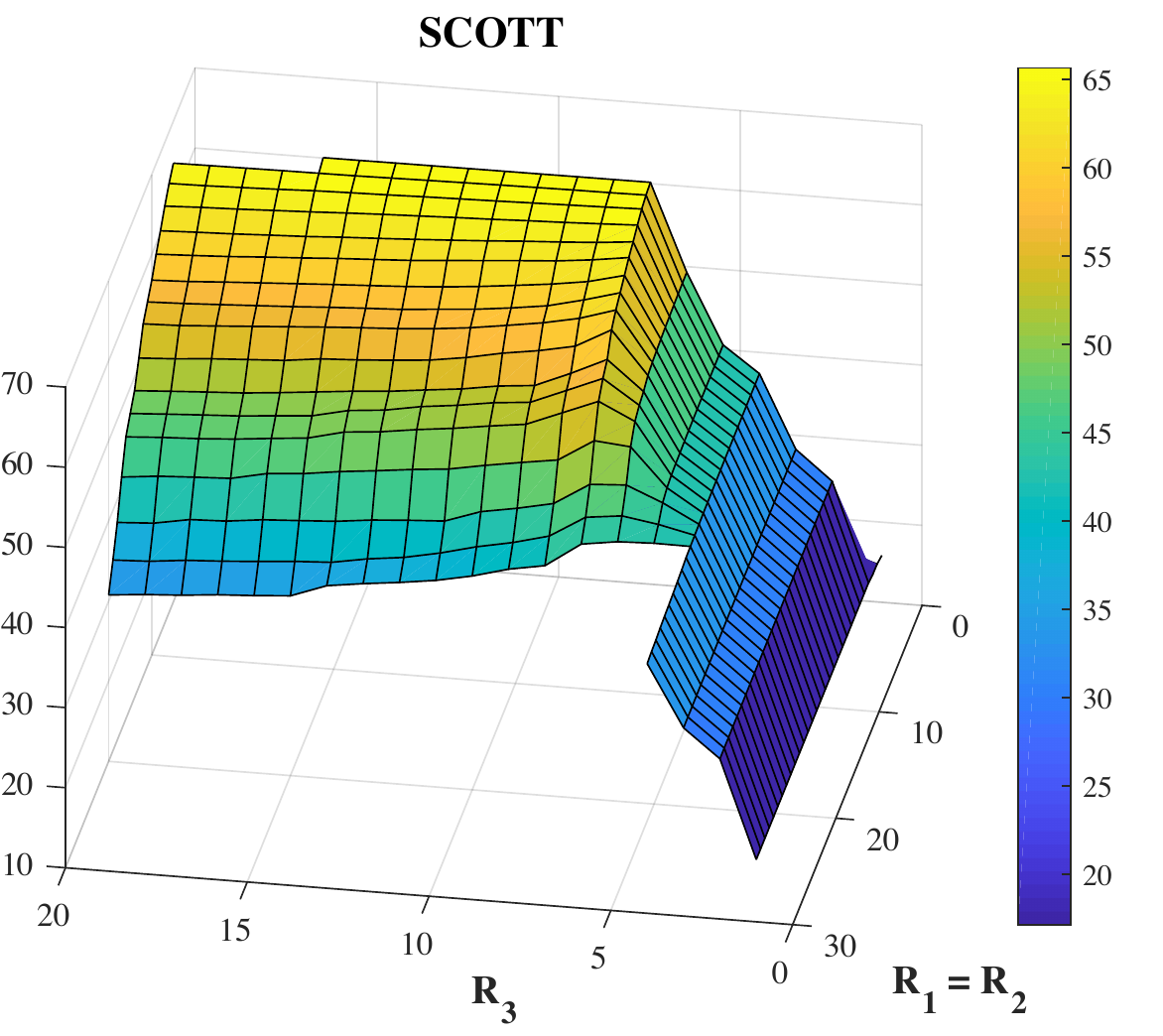}}
\end{minipage}
\caption{R-SNR as a function of the rank}
\label{fig:synth2}
\end{figure}

\subsubsection{Block tensor}\label{sec:btensor}
	
	Here, we provide an example in which the CPD of the MSI is not unique but the CP approach still achieves the correct recovery of the SRI. 
	This dataset is made of $N=6$ materials with spatial degradation ratio of $d=4$ for the HSI and Quickbird specifications for the MSI so that $K_M < N$.
Each abundance map is made of two $10\times 10$ Gaussians of width $\sigma = 4$, as in Table~\ref{tab:block}.

	\begin{table}[htb!]
\begin{center} {\footnotesize
    \begin{tabular}{|l|l|l|}
      \hline
     1\quad  & &\\
     \quad\quad 1 & & \\
      \hline
              & $\ddots$\quad &  \\
               \hline
               & & $N$ \quad \\
               & & \quad\quad $N$ \\
      \hline
    \end{tabular}}
        \end{center}
        \caption{Parcel map for blocked tensor, $N=6$}
        \label{tab:block}
    \end{table}

In this example, the tensor rank of $\tY$ is $F=12$ while the multilinear rank is $\vect{R}=(12,12,6)$. 
The CP decompositions of both MSI and HSI  are not unique, but the recoverability conditions given in Corollary~\ref{cor:recoveryNonidentifiable} are satisfied. 
This is an example of a tensor admitting a block-term decomposition, where the abundance maps in (\ref{eq:synth}) corresponding to different materials are not rank-one. 
While this is not a realistic example due to small ranks of abundance maps, it is inspired by the standard linear mixing model \cite{YokoyaGC17:hsr} with few materials.  

\begin{figure}[htb!]
\begin{minipage}[b]{0.49\linewidth}
 \centering
 {\includegraphics[height=3.6cm]{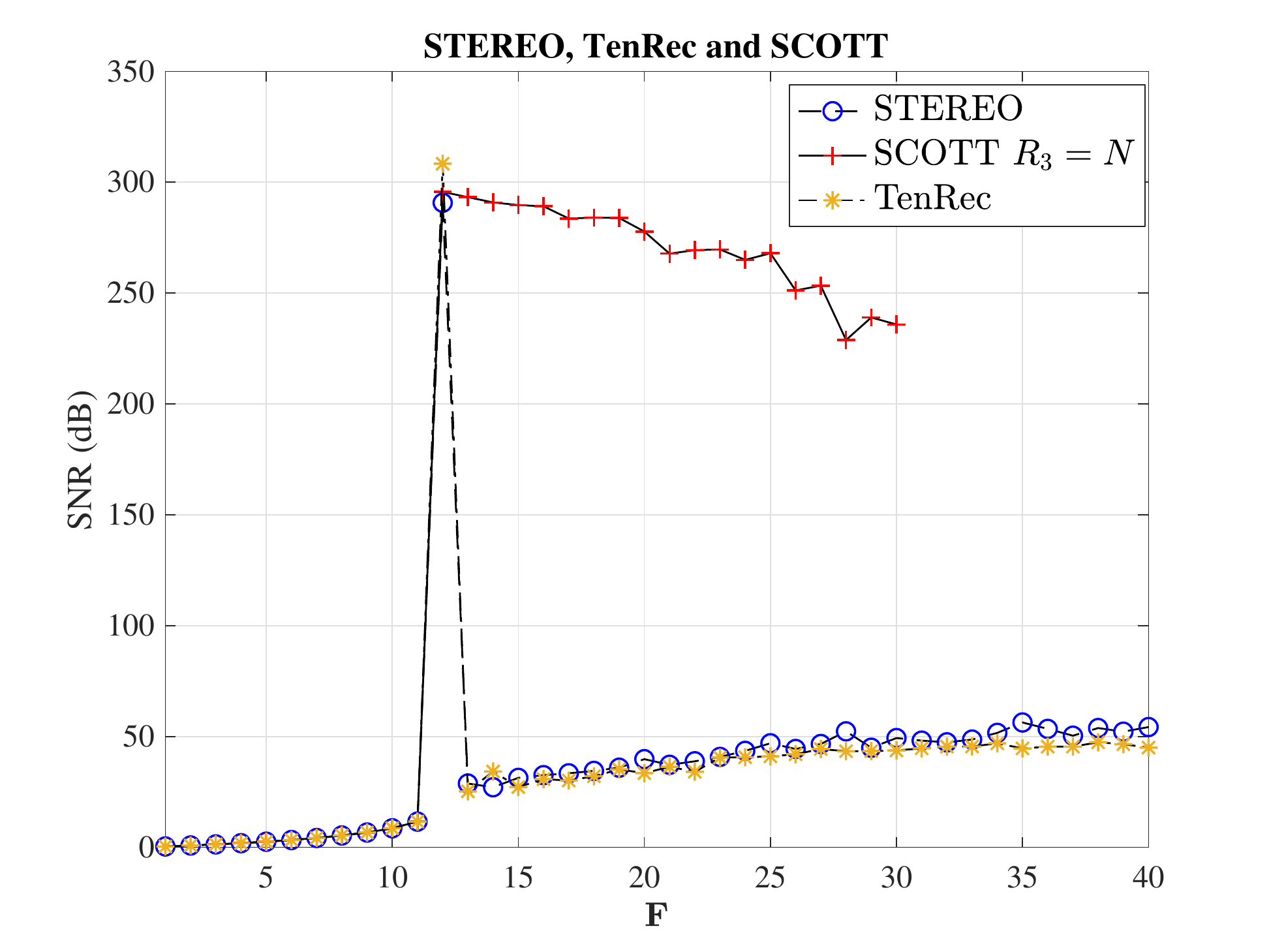}}
\end{minipage}
\hfill
\begin{minipage}[b]{0.49\linewidth}
 \centering
 {\includegraphics[height=3.6cm]{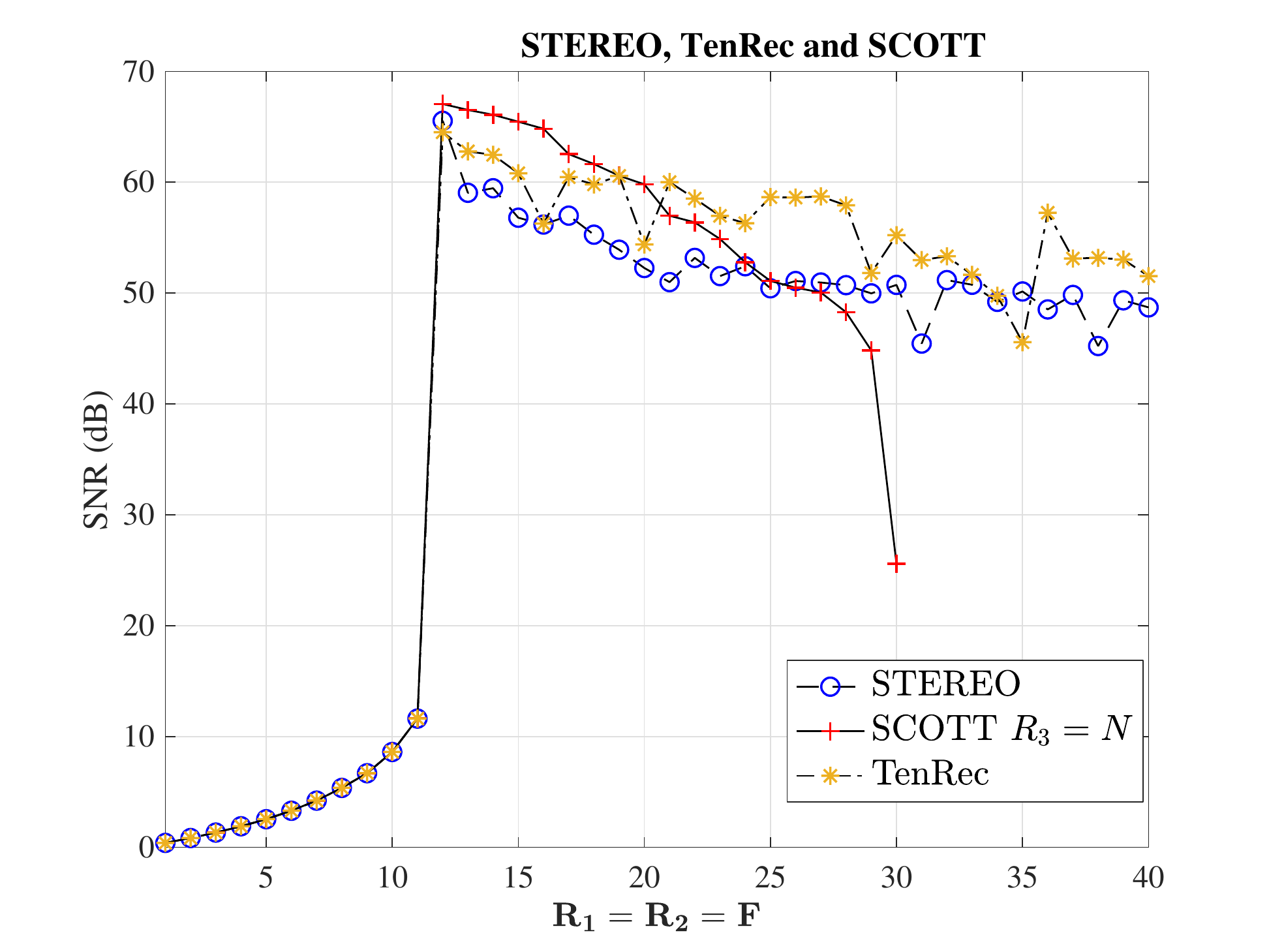}}
\end{minipage}
\caption{R-SNR as a function of the rank in the noiseless case (left) and with 35dB input SNR (right) }
\label{fig:bktens}
\end{figure}
	
In Figure~\ref{fig:bktens},  we show the R-SNR as a function of the rank for STEREO and SCOTT for the noiseless and noisy cases. 
In the noiseless case, under recoverability conditions, our Tucker-based approach provides good reconstruction for a variety of ranks and $R_3 \geq N$.
For STEREO and TenRec, we can see that even though the CP model is not identifiable, $F=12$ allows correct reconstruction of the SRI with almost the same performance as that of SCOTT for $\vect{R}=(12,12,6)$ (up to machine precision). 

This example corroborates  Corollary~\ref{cor:recoveryNonidentifiable} and shows that identifiability of the CP model (as it is formulated in \cite{KanatsoulisFSM:hsr}) is not necessary to reconstruct $\tY$ accurately,  and partial uniqueness may be sufficient.

In the noisy case, the three algorithms have almost the same performance for $R_1=R_2=F=12$.
However, for $F\geq 21$, TenRec gives better performance than SCOTT, and for $F\geq 26$, STEREO overcomes our approach.

\subsection{Choice of multilinear ranks in the presence of noise}
In Section~\ref{sec:id}, we provided a theorem for recoverability of the SRI. 
In this subsection, we show that the conditions of Theorem~\ref{thm:TuckerIdentifiabilityDeterministic} also give hints on choosing the multilinear ranks for HSR in ``signal+noise'' and semireal scenarios.

\subsubsection{Singular values of the unfoldings}
Motivated by step 1 of Algorithm~\ref{alg:hosvd}, where the factor matrices $\mU,\mV,\mW$ are computed by HOSVD of the HSI and MSI,  and by the first set of conditions in Theorem~\ref{thm:TuckerIdentifiabilityDeterministic}, we look at the singular values of $\unfold{Y}{1}_M$, $\unfold{Y}{2}_M$ and $\unfold{Y}{3}_H$.

We first consider the synthetic data from Figure~\ref{fig:synth} with $N=2$ materials, and  add white Gaussian noise to $\tY_H$ and $\tY_M$ with different SNR: 20dB, 35dB, 60dB and no noise.
In Figure~\ref{fig:svd1}, we plot the 15 first singular values of the unfoldings on a semi-log scale.

\begin{figure}[htb!]
\begin{minipage}[b]{.98\linewidth}
  \centering
  \centerline{\includegraphics[width=6.7cm]{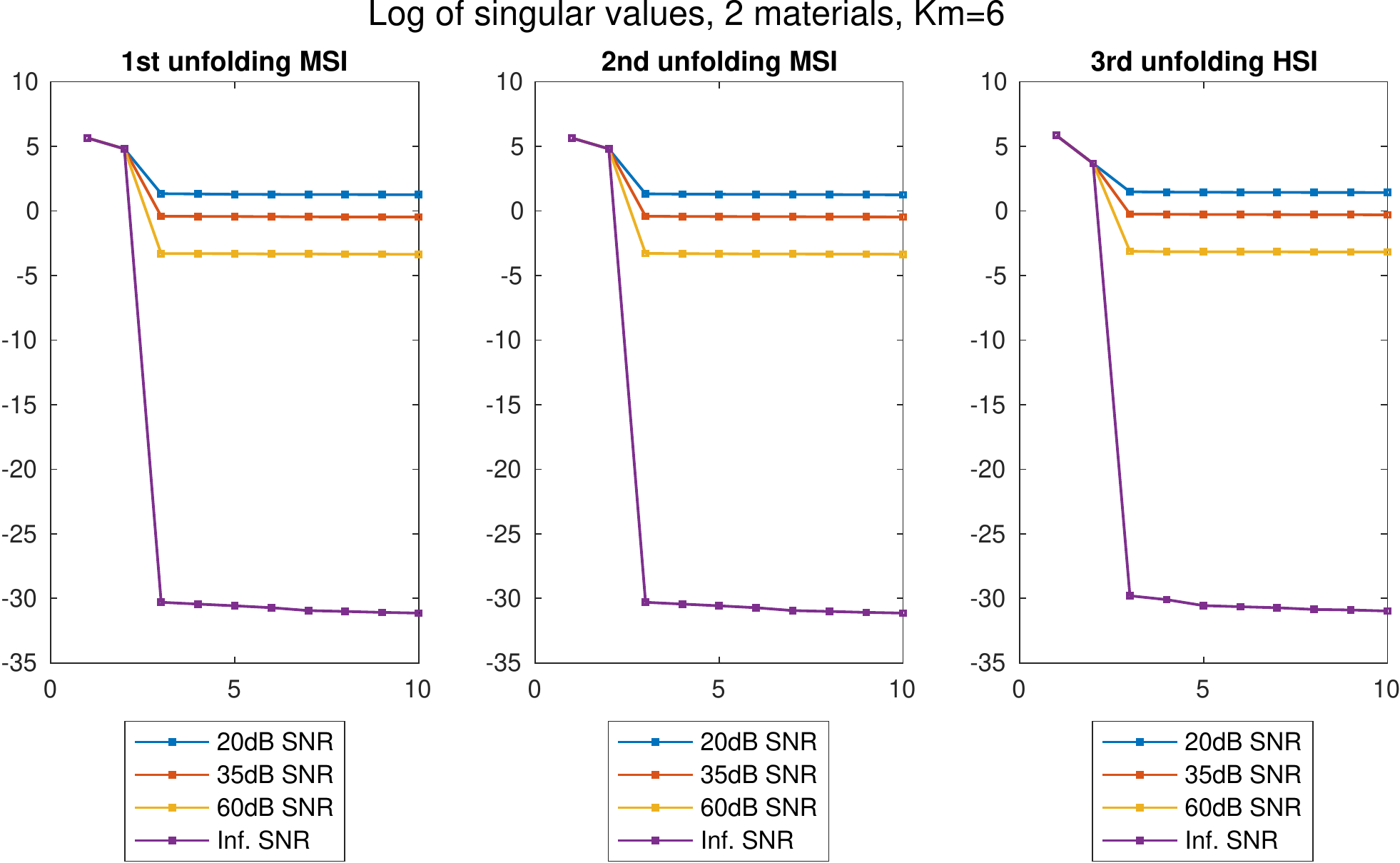}}
  \caption{Logarithm of the first 15 singular values for the three unfoldings}\medskip
  \label{fig:svd1}
\end{minipage}
\end{figure}

We can see that for all the considered noise levels, the singular values are well separable. 
The corners of the curves at singular values $(2,2,2)$ are coherent with the theoretical multilinear rank of the synthetic SRI.

We now consider the semi-real datasets Indian Pines and Salinas-A and plot the singular values of the unfoldings on a semi-log scale on Figures~\ref{fig:svd_ip} and~\ref{fig:svd_salA}.

\begin{figure}[htb!]
\begin{minipage}[b]{.98\linewidth}
  \centering
  \centerline{\includegraphics[width=8cm]{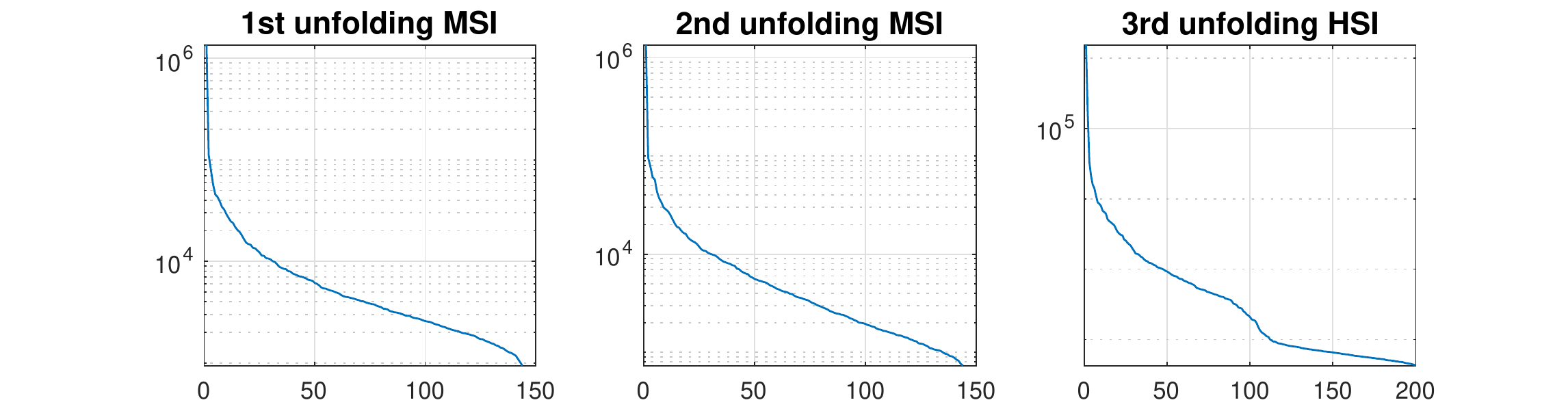}}
  \caption{Singular values for the three unfoldings, Indian Pines}\medskip
  \label{fig:svd_ip}
\end{minipage}
\end{figure}

\begin{figure}[htb!]
\begin{minipage}[b]{.98\linewidth}
  \centering
  \centerline{\includegraphics[width=8cm]{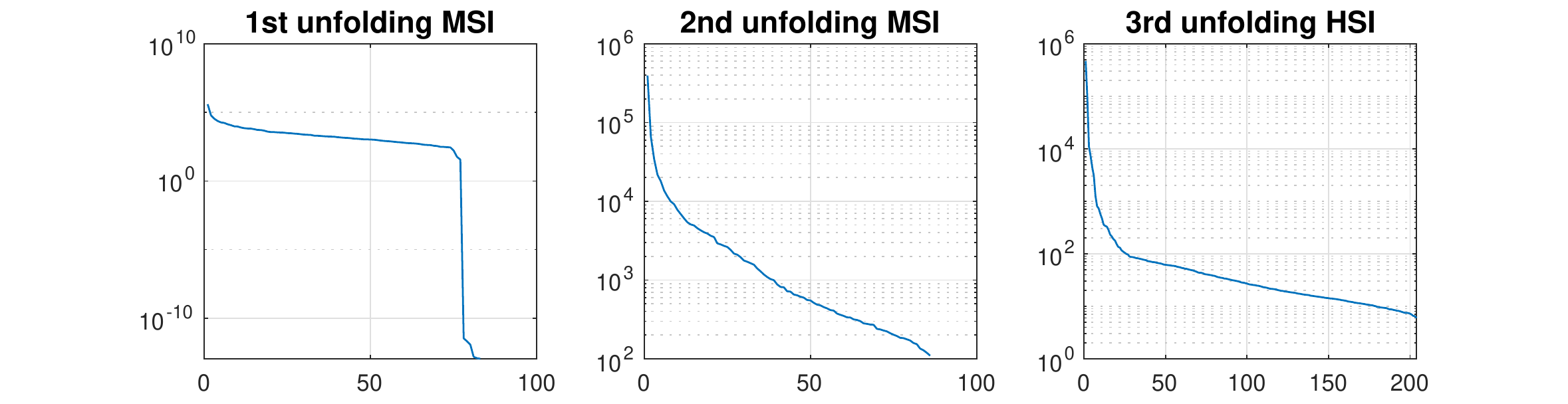}}
  \caption{Singular values for the three unfoldings, Salinas-A scene}\medskip
  \label{fig:svd_salA}
\end{minipage}
\end{figure}

In the semi-real cases,  a clear corner in the singular value curves  cannot be found, 
 because these examples do not correspond to  a ``low-rank signal+noise'' scenario, contrary to the case of synthetic data.
Moreover, the HSI and MSI are not necessarily low-rank: hence, the Tucker approach only performs a \emph{low-rank approximation} of the data.
Hence, the SVD of the unfolding does not provide as much information as for the synthetic case, in which the groundtruth data are explicitly designed to be low-rank.

\subsubsection{Influence on the reconstruction error}
Next,  we consider the R-SNR and cost function $f_T$ as functions of the multilinear rank.
We run SCOTT for the ranks $R_1 = R_2$ in $[10:50]$ and $R_3$ in $[2:25]$ for which the recoverability condition holds (see Section~\ref{sec:id}), and two semi-real datasets: Indian Pines and Salinas-A scene.
The results are shown in Figures~\ref{fig:R2f_IP} and~\ref{fig:R2f_Sal}, respectively.
	
\begin{figure}[htb!]
\begin{minipage}[b]{.49\linewidth}
 \centering
 \centerline{\includegraphics[width=3.9cm]{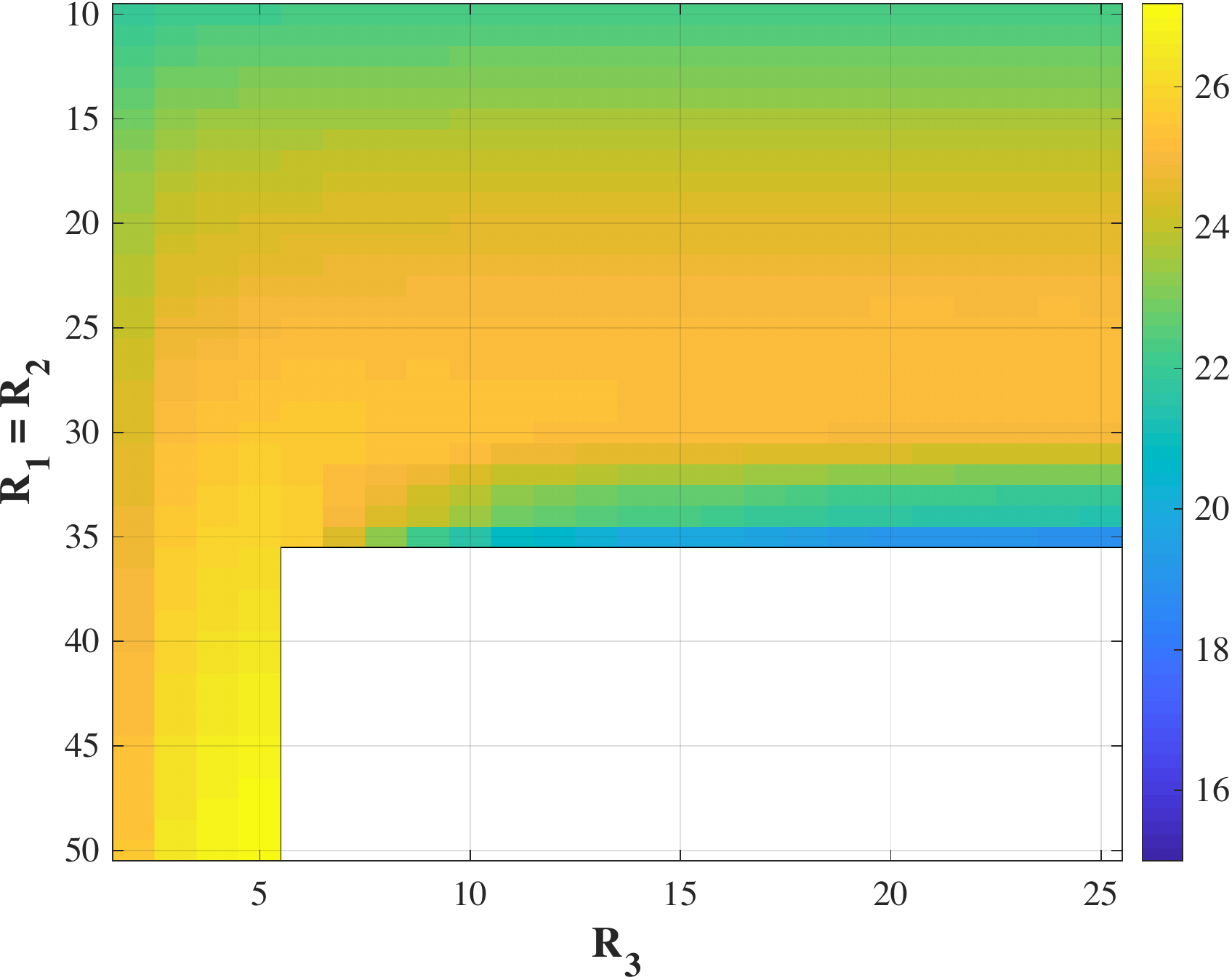}}
\end{minipage}
\hfill
\begin{minipage}[b]{0.49\linewidth}
 \centering
 \centerline{\includegraphics[width=3.9cm]{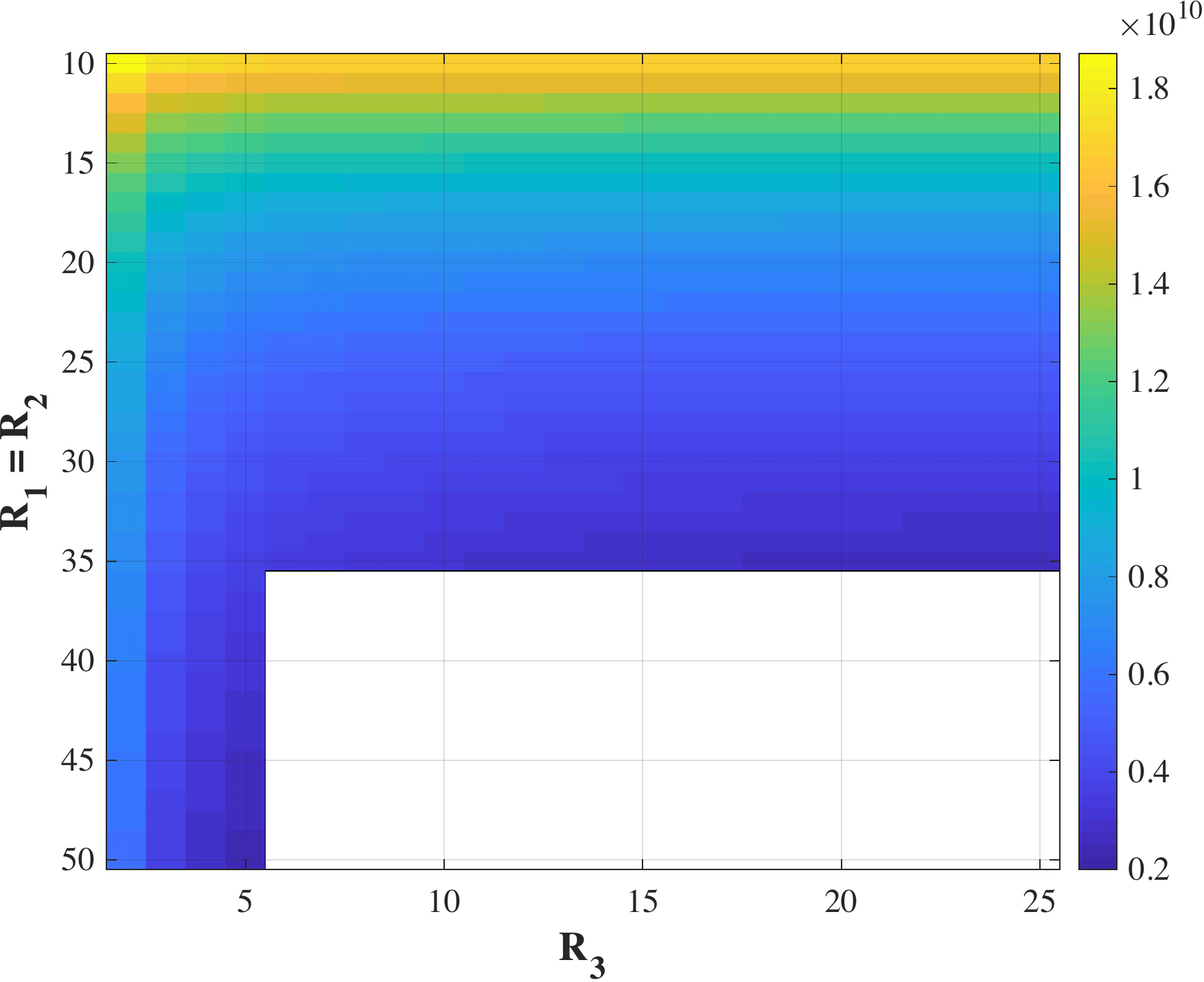}}
\end{minipage}
 \vspace{-0.4cm}
\centering
\caption{\,\!\!\!R-SNR (left) and $f_T$\! (right) as functions of $R_1$ and $R_3$, Indian~Pines}
\label{fig:R2f_IP}
\end{figure}

\begin{figure}[htb!]
\begin{minipage}[b]{.49\linewidth}
 \centering
 \centerline{\includegraphics[width=3.85cm]{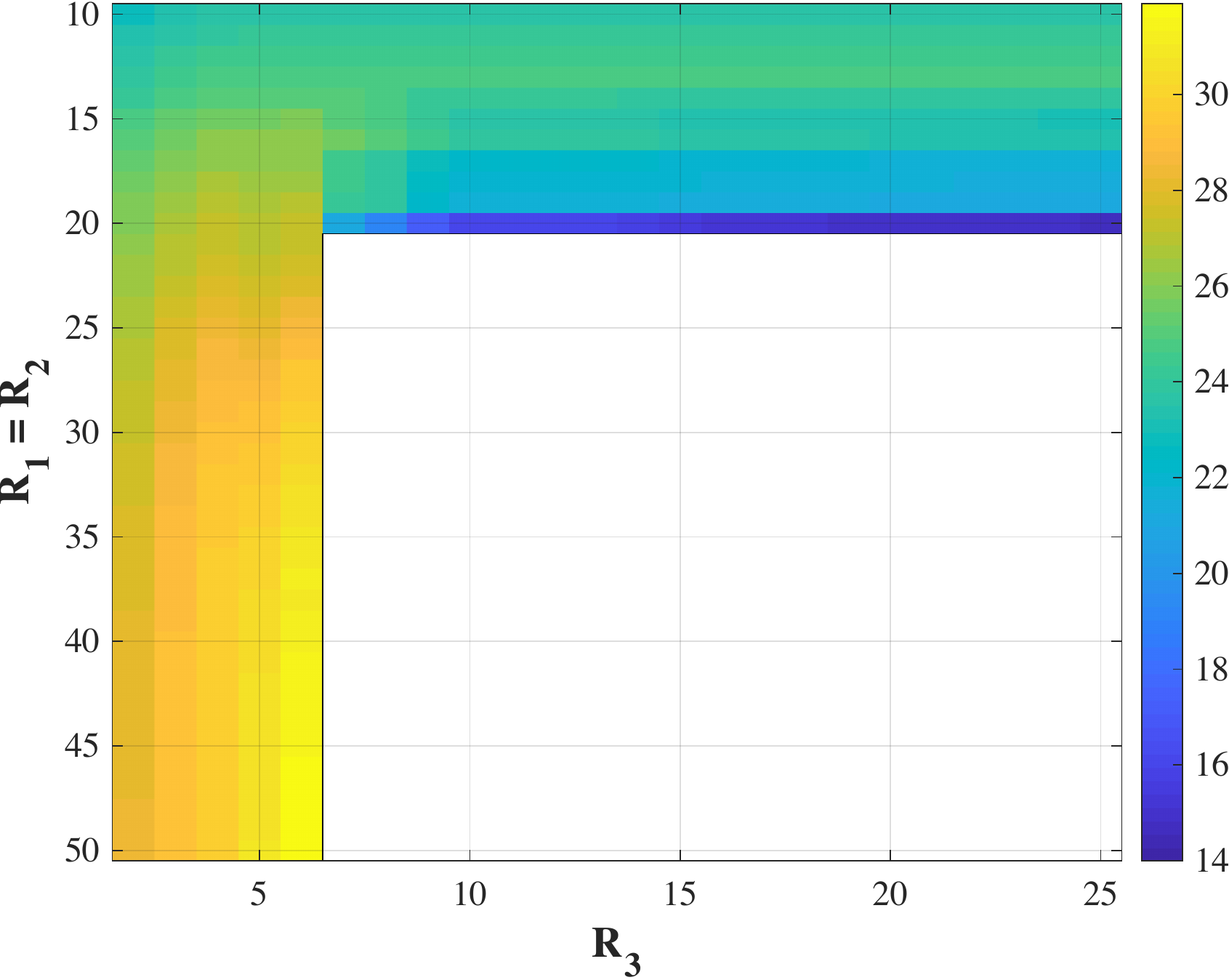}}
\end{minipage}
\hfill
\begin{minipage}[b]{0.49\linewidth}
 \centering
 \centerline{\includegraphics[width=4cm]{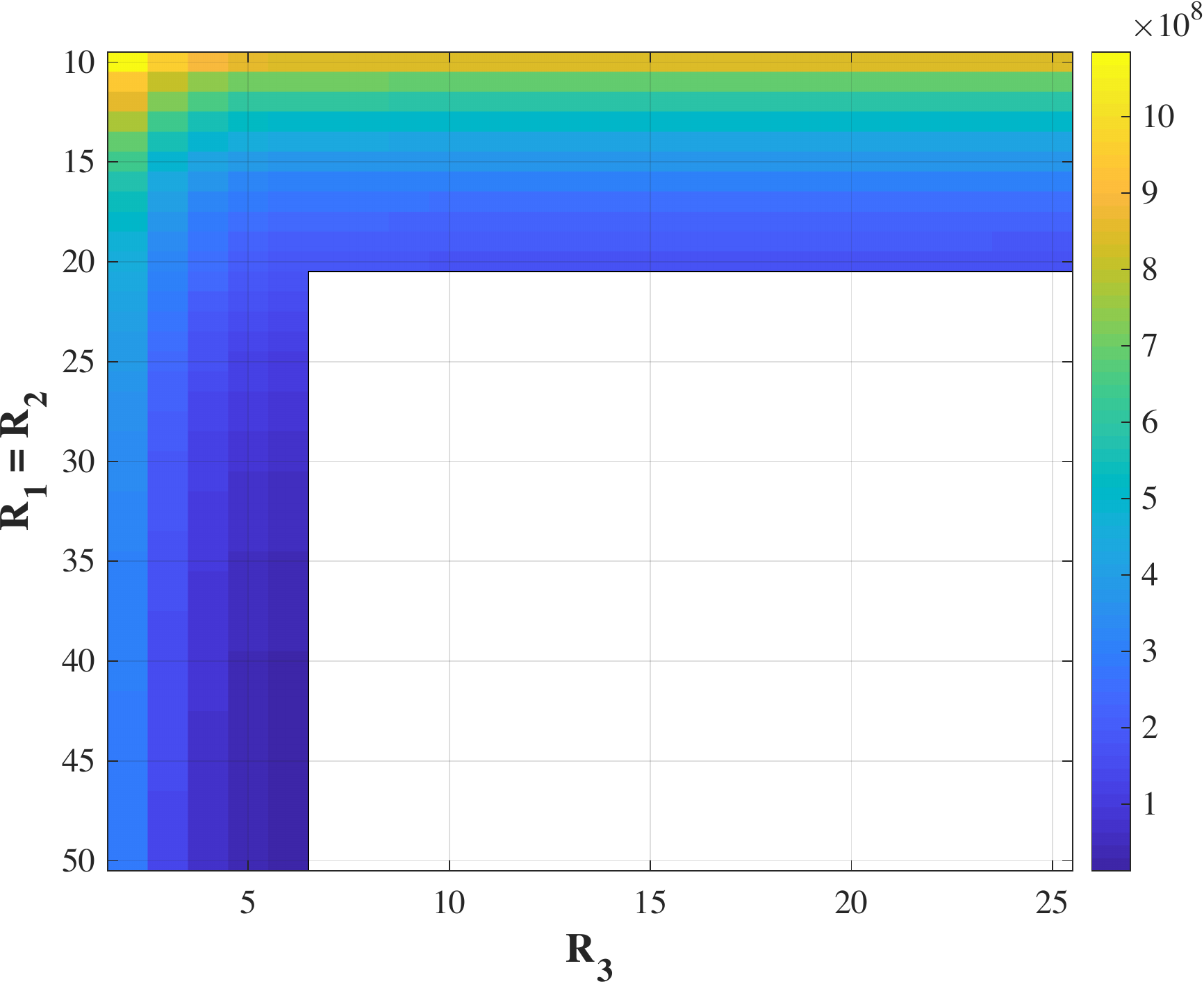}}
\end{minipage}
 \vspace{-0.4cm}
\centering
\caption{\,\!\!\!R-SNR (left) and $f_T$ (right) as functions of $R_1$ and $R_3$, Salinas-A}
\label{fig:R2f_Sal}
\end{figure}

While the cost function decreases as $R_1$ and $R_3$ increase, the best reconstruction error (given by R-SNR)  is achieved in one of the two recoverability subregions in Fig.~\ref{fig:identifiability_region}: ($a$) ($R_3\geq K_M$ and $R_1 \leq I_H$) and ($b$) ($R_3\leq K_M$ and $R_1 \geq I_H$). 
For subregion ($b$), the best performance is achieved when $R_3 = K_M$ and $R_1$ as large as possible, while for subregion ($a$), we notice a sharp drop of the R-SNR near $R_1=I_H$.

The drop of the performance in subregion (a) can be explained by looking at the condition number of the matrix $\matr{X}^{\T}\matr{X}$ that is used to compute the core tensor $\widehat{\tens{G}}$.
For the subregion (a), due to properties of Kronecker products \cite[Theorems 13.12 and 13.16]{Laub04:siam}, we have that
\[
\begin{split}
&\cond{\matr{X}^{\T}\matr{X}}  := \frac{\sigma_{max}(\matr{X}^{\T}\matr{X})}{\sigma_{min}(\matr{X}^{\T}\matr{X})} \\
&=
 \frac{\lambda\sigma^{2}_{max}(\matr{P}_M \widehat{\matr{W}}) + \sigma^2_{max}(\matr{P}_1 \widehat{\matr{U}})  \sigma^2_{max} (\matr{P}_2 \widehat{\matr{V}})}{\sigma^2_{min}(\matr{P}_1 \widehat{\matr{U}})  \sigma^2_{min} (\matr{P}_2 \widehat{\matr{V}})}.
\end{split}
\]
Note that $\sigma_{max}(\matr{P}_M \widehat{\matr{W}})$ does not decrease when we increase $R_3$ and $R_3 \le K_M$.
Hence we can get a lower bound on $\cond{\matr{X}^{\T}\matr{X}}$ by setting $R_3 = K_M$.

\begin{figure}[htb!]
\begin{minipage}[b]{0.49\linewidth}
 \centering
 {\includegraphics[height=2.5cm]{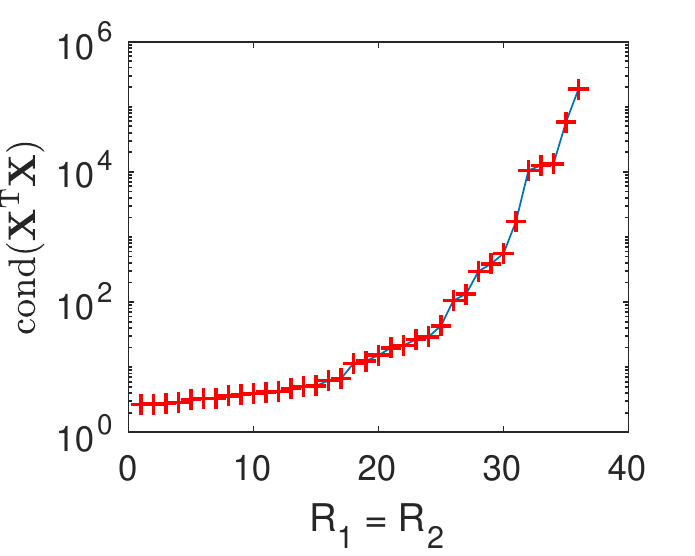}}
\end{minipage}
\hfill
\begin{minipage}[b]{0.49\linewidth}
 \centering
 {\includegraphics[height=2.5cm]{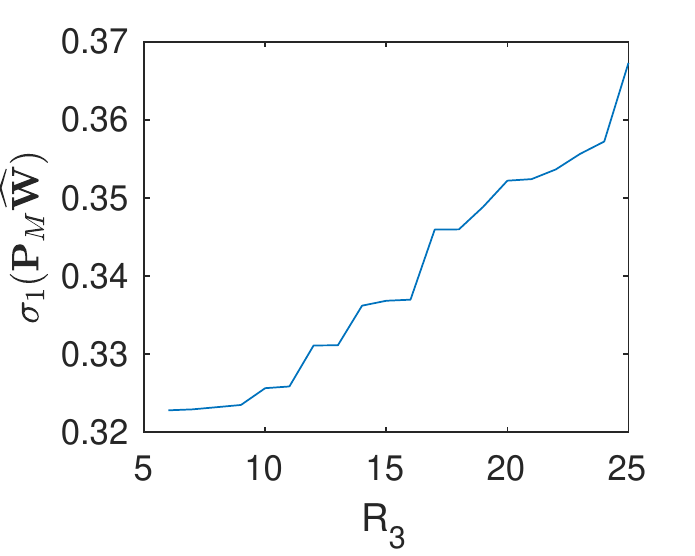}}
\end{minipage}
\caption{Left: $\log(\cond{\matr{X}^{\T}\matr{X}})$ depending on $R_1=R_2$ for $R_3 = K_M$; right:  $\sigma_1(\matr{P}_M \widehat{\matr{W}})$, Indian~Pines}
\label{fig:cond_IP}
\end{figure}

In Figure~\ref{fig:cond_IP}, for the Indian Pines dataset we plot on a semi-log scale the lower bound $\cond{\matr{X}^{\T}\matr{X}}$ as functions of $R_1 = R_2$,
for $R_3 = 6$ as well as $\sigma_{max}(\matr{P}_M \widehat{\matr{W}})$; since the latter almost does not change, the lower bound is tight.
In Figure~\ref{fig:cond_IP}, we see that there is a highest relative increase of the condition number around $R_1=R_2 = 32$, 
which coincides with the point of the performance drop in Figure~\ref{fig:R2f_IP}.
Similar behaviour can be observed for the Salinas dataset on Figure~\ref{fig:cond_sal}. 

\begin{figure}[htb!]
\begin{minipage}[b]{0.49\linewidth}
 \centering
 {\includegraphics[height=2.5cm]{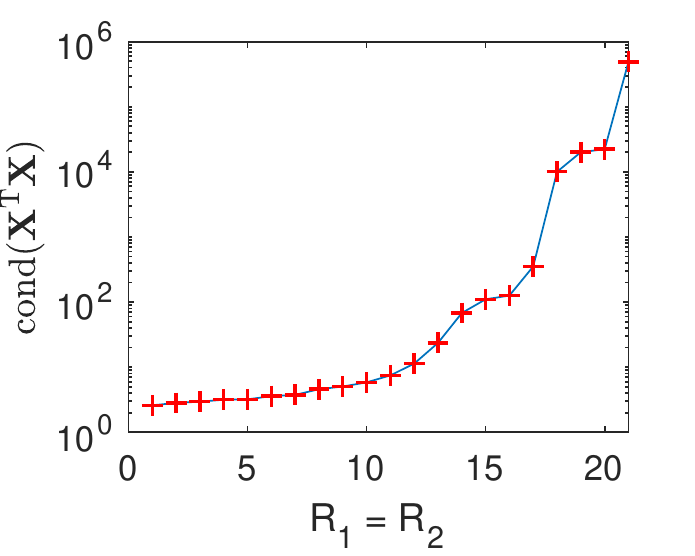}}
\end{minipage}
\hfill
\begin{minipage}[b]{0.49\linewidth}
 \centering
 {\includegraphics[height=2.5cm]{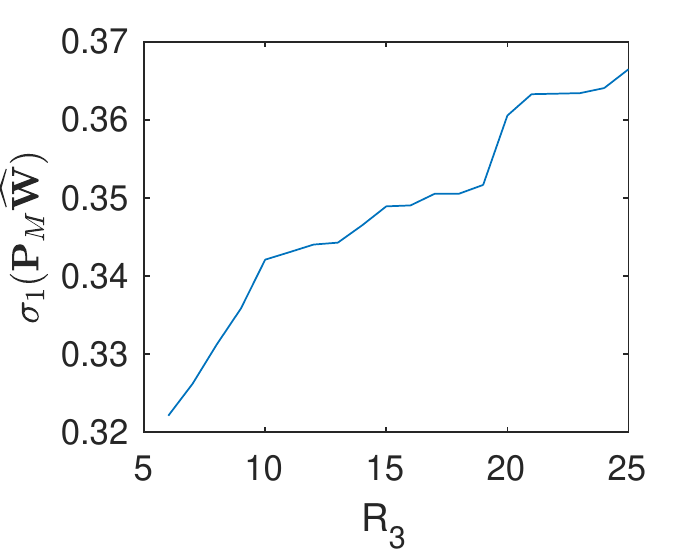}}
\end{minipage}
\caption{Left: $\log(\cond{\matr{X}^{\T}\matr{X}})$ depending on $R_1=R_2$ for $R_3 = K_M$; right:  $\sigma_1(\matr{P}_M \widehat{\matr{W}})$, Salinas-A scene}
\label{fig:cond_sal}
\end{figure}

All in all, based on the above examples, we can conclude  that if we are in subregion (b), the $R_3$ should be taken as large as possible ($R_3 = K_M$), while in the subregion (a) $R_1$, $R_2$ should be taken as large as possible while maintaining the condition number to a reasonable value.


\subsection{Recovery of spectral signatures}

Since correct recovery of spectral signatures is quite important for further processing of hyperspectral images, we would like see whether  \OurAlgo is able to do that.
We consider the Indian Pines dataset, where  groundtruth data  (see Fig. \ref{fig:gt_ip})  is available, splitting the image into 16 regions.
We will consider three representative ranks: $[40,40,6]$, $[30,30,16]$,  and $[24,24,25]$, and compare them to STEREO ($F=100$).

\begin{figure}[htb]
\begin{minipage}[b]{1.0\linewidth}
  \centering
  \centerline{\includegraphics[trim = 1.5cm 1cm 1cm 1cm, clip, width=5.6cm]{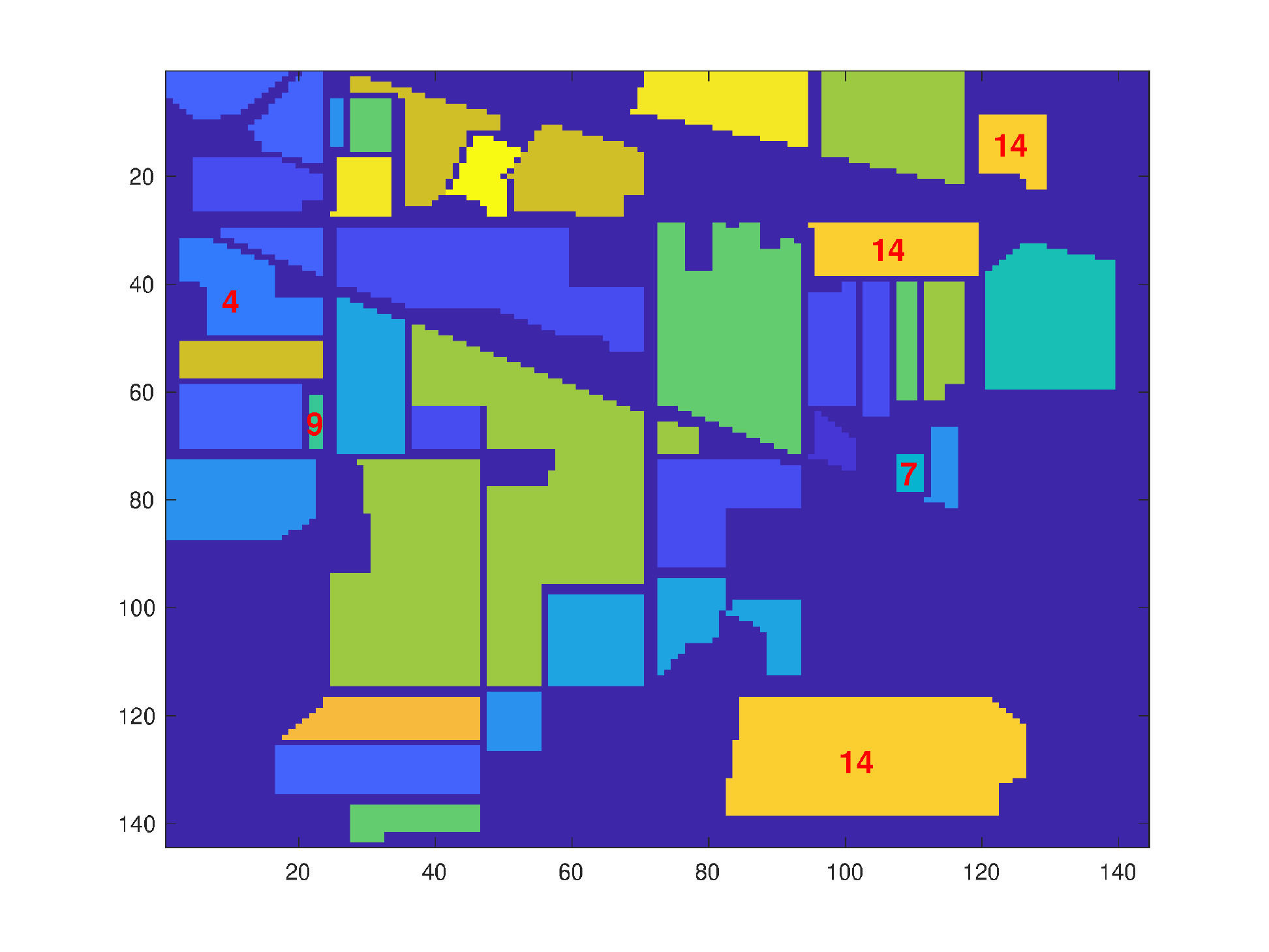}}
\end{minipage}
\vspace{-0.7cm}
\caption{Groundtruth image for Indian Pines dataset. Materials 4,7,9,14 are marked in red.}
\label{fig:gt_ip}
\end{figure}

We do not perform a proper hyperspectral unmixing, and compute the spectral signatures by averaging across the regions. 
We selected four representative signatures corresponding to materials 4,7,9 and 14, which are plotted in Figure~\ref{fig:spec1}. 
Note that materials 7 and 9 are scarce in the groundtruth SRI (resp. 28 and 20 pixels), whereas materials 4 and 14 are more abundant (resp. 237 and 1265 pixels).

\begin{figure}[htb!]
\begin{minipage}[b]{1.0\linewidth}
  \centering
  \centerline{\includegraphics[width=5.2cm]{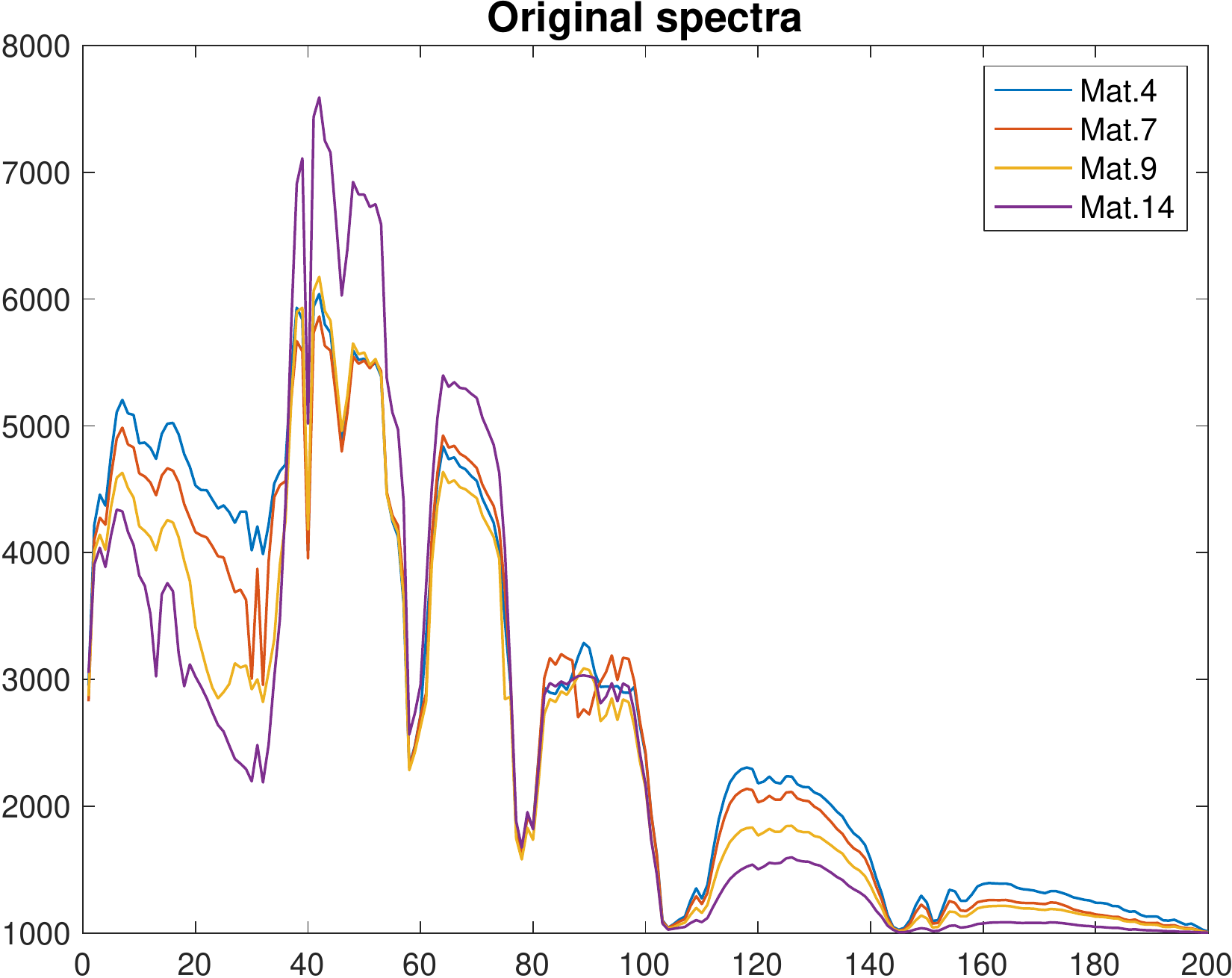}}
\end{minipage}
\label{fig:spec1}
\vspace{-0.5cm}
\caption{Original spectral signature for materials 4,7,9 and 14}
\end{figure}

In Figure~\ref{fig:exp2_allres} we plot relative errors of the reconstruction of spectra by different methods.
As expected, for materials 7 and 9, the discrepancy between the original spectra and the spectra obtained from estimated SRI is bigger than for materials 4 and 14.
This can be explained by the scarcity of sources 7 and 9  compared to sources 4 and 14.

\begin{figure}[htb]
\begin{minipage}[b]{1.0\linewidth}
  \centering
  \includegraphics[width=5.5cm]{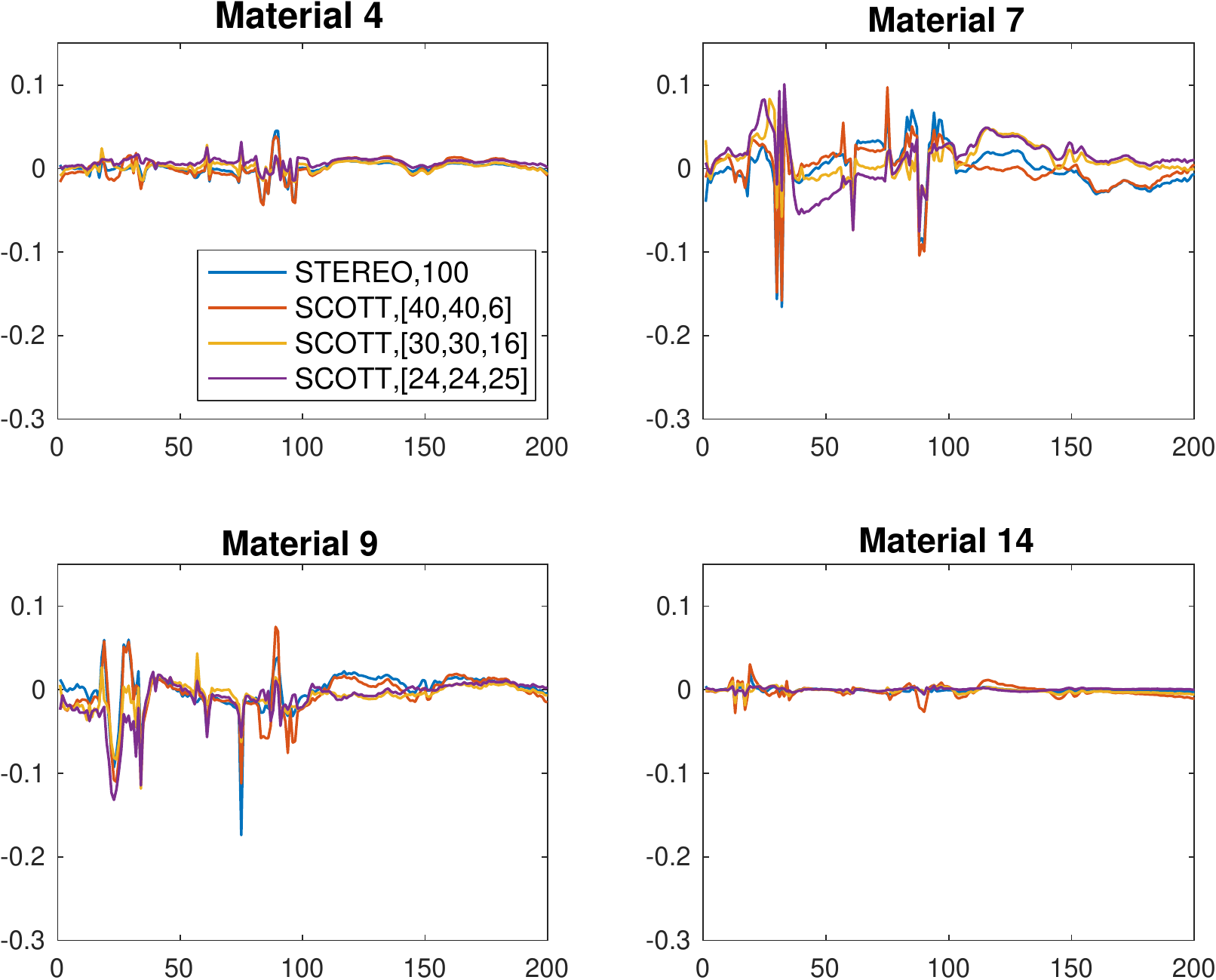}
   \caption{Residual errors for the three considered ranks and four materials}\medskip
     \label{fig:exp2_allres}
  \end{minipage}
\end{figure}

In Figure~\ref{fig:exp2_allspec}, we have a closer look at the spectra at spectral bins 80 to 100.
We can see that for abundant materials (4 and 14), all the algorithms estimate well the spectra.
For the scarce materials it is important to choose the rank large enough, in particular $R_3 = 16$ and $R_3 = 25$ yield better reconstruction than $R_3 =6$, and also than STEREO, even with $F=100$.

\begin{figure}[htb]
\begin{minipage}[b]{1.0\linewidth}
  \centering
  \centerline{\includegraphics[width=5.5cm]{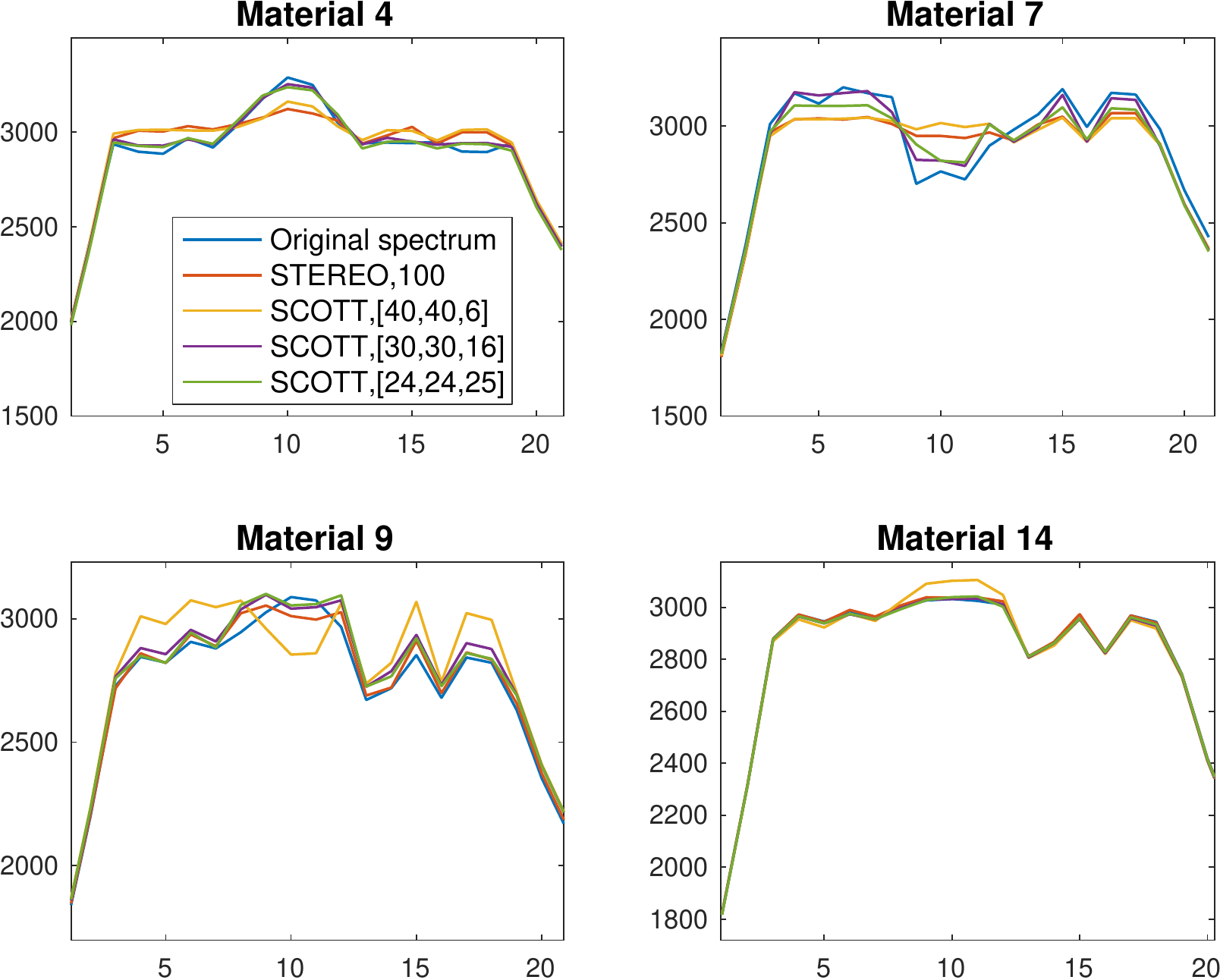}}
\end{minipage}
\caption{Materials at spectral bins  80 to  100.  Groundtruth (\emph{black}), \OurAlgo $(40,40,6)$ (\emph{red}), \OurAlgo $[30,30,16]$ (\emph{yellow}), \OurAlgo $(24,24,25)$ (\emph{purple}), STEREO $F=100$ (\emph{green}).}
\label{fig:exp2_allspec}
\end{figure}

\section{CONCLUSION}
In this paper, we proposed a novel coupled Tucker model for hyperspectral superresolution.
We showed that the model is \emph{recoverable}, that is, it allows for an unique recovery of the SRI for a wide range of multilinear ranks. We proposed two very simple SVD-based algorithms that can be used for the super-resolution problem,  for known and unknown degradation operators, and for the case of pansharpening.
The algorithms are very fast, but produce the results that are comparable with the CP-based approaches.
This work opens new perspectives on using various tensor factorizations for hyperspectral super-resolution.
Still several interesting questions remain,  for example,  how to enlarge the recoverability range for the multilinear rank.
Estimating the multilinear rank of the Tucker decomposition still remains an open problem; the question of optimal splitting the data into non-overlapping subtensors in B-SCOTT also needs to be further investigated.

\appendices
\section{Solving normal equations as generalized Sylvester equations} \label{app:1}

Equation~(\ref{eq:ls_core}) can be seen as a generalized Sylvester equation of the form 
\begin{equation}\label{eq:app}
\mA\widehat{\matr{G}}\mB + \mC\widehat{\matr{G}}\matr{D} = \matr{E},
\end{equation}
where $\matr{G}$ is an unfolding of $\widehat{\tens{G}}$.

We propose two options for converting~(\ref{eq:ls_core}) into~(\ref{eq:app}).
In the first case, $\widehat{\matr{G}} = \unfold{\tG}{3} \in\RR^{R_1R_2\times R_3}$,
\begin{align*}
& \matr{A} = \left(\matr{U}^{\T}\matr{P}_1^{\T}\matr{P}_1\mU\right)\kron\left( \matr{V}^{\T}\matr{P}_2^{\T}\matr{P}_2\mV\right), \quad \matr{B} = \matr{I}_{R_3}, \\
& \matr{C} = \matr{I}_{R_1R_2},  \quad \matr{D} = \lambda\left( \matr{W}^{\T}\matr{P}_M^{\T}\matr{P}_M\mW\right),
\end{align*}
and $\matr{E}\in\RR^{R_1R_2\times R_3}$ is a matricization of $\matr{X}^{\T}\vect{z}$.

In the second case, $\widehat{\matr{G}} = {\unfold{\tG}{1}}^{\T} \in\RR^{R_1\times R_2R_3}$,
\begin{align*}
& \matr{A} = \matr{U}^{\T}\matr{P}_1^{\T}\matr{P}_1\mU, \quad \matr{B} = \matr{I}_{R_3}\kron\left( \matr{V}^{\T}\matr{P}_2^{\T}\matr{P}_2\mV\right), \\
& \matr{C} = \matr{I}_{R_1}, \quad \matr{D} = \lambda\left( \matr{W}^{\T}\matr{P}_M^{\T}\matr{P}_M\mW\right)\kron \matr{I}_{R_2},
\end{align*}
and $\matr{E}\in\RR^{R_1\times R_2R_3}$ is a matricization of $\matr{X}^{\T}\vect{z}$.

The two options are equivalent and the fastest one is chosen according to the multilinear rank.
As a rule of thumb, we decide to choose the first option in subregion $(a)$ of Figure~\ref{fig:identifiability_region} and the second option in subregion $(b)$.
The complexity for solving the generalized Sylvester equation~(\ref{eq:app}) is thus $O(m^3 +  n^3)$ flops for $\widehat{\matr{G}} \in \RR^{m\times n}$ if fast solvers, such as Hessenberg-Schur or Bartels-Stewart methods \cite{BartelsS72:syl}, \cite{GolubNV79:syl},  \cite{Simo16:sirev},  are used.

\section{Degradation matrices} \label{app:2}

Here, we explain in details how the degradation matrices are constructed. 
For this appendix, we consider that $\matr{P}_1 = \matr{P}_2$.
As in \cite{KanatsoulisFSM:hsr}, $\matr{P}_1$ is constructed as $\matr{P}_1  = \matr{S}_1\matr{T}_1$, where $\matr{T}_1$ is a blurring matrix and $\matr{S}_1$ is a downsampling matrix.

The blurring matrix is constructed from a Gaussian blurring kernel $\matr{\phi} \in\RR^{q\times 1}$ (in our case, $q=9$) with a standard deviation $\sigma$. For $m=1,\ldots,q$ and $m' = m- \left \lceil{\frac{q}{2}}\right \rceil $, we have
\begin{equation*}
\phi (m) = \frac{1}{\sqrt{2\pi\sigma^2}}\exp\left( \frac{-m'^2}{2\sigma^2}\right).
\end{equation*}
Thus, $\matr{T}_1 \in \RR^{I\times I}$ can be seen as

\begin{equation*}
\matr{T}_1 = \left[\begin{smallmatrix} \phi (\left \lceil{\frac{q}{2}}\right \rceil ) & \ldots & \phi (q) & 0 & \ldots & 0 \\ 
\vdots & \ddots & & \ddots & \ddots & \vdots\\
\phi (1) & & \ddots & & \ddots & 0 \\
0 & \ddots & & \ddots & & \phi (q)\\
\vdots & \ddots & \ddots & &\ddots & \vdots \\
0 & \ldots  & 0 & \phi (1) & \ldots & \phi (\left \lceil{\frac{q}{2}}\right \rceil ) \end{smallmatrix}\right].
\end{equation*}

The downsampling matrix $\matr{S}_1 \in \RR^{I_H\times I}$, with downsampling ratio $d$, is made of $I_H$ independant rows such that for $i=1,\ldots, I_H$, $(\matr{S}_1)_{i,2+(i-1)d} = 1$ and the other coefficients are zeros.

The spectral degradation matrix $\matr{P}_M \in \RR^{K_M\times K}$ is a selection-averaging matrix, Each row represents a spectral range in the MSI; coefficients are set to ones for common bands with the SRI, and zeros elsewhere. 
The coefficients are averaged per-row.
Below, we give an example of a $2\times 6$ matrix:
\begin{equation*}
\begin{bmatrix} 0 & \frac{1}{3} & \frac{1}{3} & \frac{1}{3} & 0 & 0 \\
0 & 0 & 0 & 0 & \frac{1}{2} & \frac{1}{2} \end{bmatrix}.
\end{equation*}



\bibliographystyle{IEEEbib}
\bibliography{HSR}

\end{document}